\documentclass[review,hidelinks,onefignum,onetabnum]{siamart220329}

\usepackage{graphicx}
\usepackage{multirow}
\usepackage{amsmath,amssymb,amsfonts}
\usepackage{mathrsfs}
\usepackage[title]{appendix}
\usepackage{xcolor}
\usepackage{textcomp}
\usepackage{manyfoot}
\usepackage{booktabs}
\usepackage{algorithm}
\usepackage{algorithmicx}
\usepackage{algpseudocode}
\usepackage{listings}

\usepackage[english]{babel}
\usepackage[latin1]{inputenc}
\usepackage{enumerate}
\usepackage{color}
\usepackage[T1]{fontenc}
\usepackage{epstopdf}
\usepackage{subfigure}
\usepackage{dsfont}
\usepackage[T1]{fontenc}
\usepackage{amsmath}
\usepackage{mathtools}
\usepackage{amstext}
\usepackage{amssymb}
\usepackage{mathrsfs}
\usepackage{mathtools}
\usepackage{tikz}
\usetikzlibrary{arrows,positioning}
\usepackage{pifont}

\usepackage{tikz}
\usetikzlibrary{shapes,arrows}
\tikzstyle{block} = [draw, fill=white, rectangle,
minimum height=3em, minimum width=6em]
\tikzstyle{sum} = [draw, fill=white, circle, node distance=1cm]
\tikzstyle{input} = [coordinate]
\tikzstyle{output} = [coordinate]
\tikzstyle{pinstyle} = [pin edge={to-,thin,black}]

\def\D{{\mathcal D}}

\def\F{{\mathcal F}}
\def\H{{\mathcal H}}

\def\C{{\mathcal C}}

\def\P{{\mathcal P}}

\def\S{{\mathcal S}}
\def\cE{\mathbb{E}}

\def\bpi{{\boldsymbol \pi}}
\def\bgamma{{\boldsymbol \gamma}}

\def\Mar{{\mathsf{Mar}}}
\def\Target{{\mathsf{Target}}}
\def\Inv{{\mathsf{Inv}}}

\def\sTr{{\mathsf{Tr}}}

\def\sX{{\mathsf X}}

\def\sA{{\mathsf A}}

\def\sM{{\mathsf M}}
\def\sE{{\mathsf E}}

\def\sU{{\mathsf U}}

\def\cL{\mathcal{L}}
\def\R{\mathbb{R}}
\def\rC{\mathbb{C}}
\def\N{{\mathcal N}}

\def\rP{\mathbb{P}}
\def\Id{\mathbb{Id}}

\def\diam{\mathop{\rm diam}}
\def\dim{\mathop{\rm dim}}
\def\diag{\mathop{\rm diag}}
\def\tr{\mathbb{\rm Tr}}

\def\T{\mathbb{\rm T}}

\def\L{\mathbb{\rm L}}

\def\cC{\mathop{\rm C}}

\def\argmin{\mathop{\rm arg\, min}}

\newtheorem{example}{Example}

\allowdisplaybreaks

\begin{document}

\sloppy
\title{Quantum Markov Decision Processes: General Theory, Approximations, and Classes of Policies}
\author{Naci Saldi, Sina Sanjari, and Serdar Y\"uksel\thanks{Naci Saldi is with Department of Mathematics at Bilkent University, Ankara, Turkey, Email: \{\email{naci.saldi@bilkent.edu.tr}\}. Sina Sanjari is with the Department of Mathematics and Computer Science at Royal Military College, Kingston, ON, Canada, Email: \{\email{sanjari@rmc.ca}\}.
Serdar Y\"uksel is with the Department of Mathematics and Statistics, Queen's University, Kingston, ON, Canada,
Email: \{\email{yuksel@queensu.ca}}\}.}
\maketitle

\begin{abstract}
In this paper, the aim is to develop a quantum counterpart to classical Markov decision processes (MDPs). Firstly, we provide a very general formulation of quantum MDPs with state and action spaces in the quantum domain, quantum transitions, and cost functions. Once we formulate the quantum MDP (q-MDP), our focus shifts to establishing the verification theorem that proves the sufficiency of Markovian quantum control policies and provides a dynamic programming principle.  Subsequently, a comparison is drawn between our q-MDP model and previously established quantum MDP models (referred to as QOMDPs) found in the literature. Furthermore, approximations of q-MDPs are obtained via finite-action models, which can be formulated as QOMDPs. Finally, classes of open-loop and classical-state-preserving closed-loop  policies for q-MDPs are introduced, along with structural results for these policies. In summary, we present a novel quantum MDP model aiming to introduce a new framework, algorithms, and future research avenues. We hope that our approach will pave the way for a new research direction in discrete-time quantum control.
\end{abstract}



\section{Introduction}\label{sec0}

In this paper, we develop a novel mathematical formulation of quantum MDPs with various quantum control policies. While doing so, we place the classical MDP theory and recent studies on quantum MDPs in the general framework presented.

For the classical setup, MDPs (and their partially observable and decentralized counterparts) have proven to be versatile models for mathematical analysis and applied relevance for a variety of areas such as engineering, economics, and biology \cite{HeLa96}. In MDP theory, the primary objective revolves around selecting an optimal control policy, typically geared towards minimizing a cost function, with common choices being the discounted cost, average cost, or finite horizon cost criteria. In this paper, we focus on the discounted cost scenario. There exist two well-established approaches for solving such problems, namely the dynamic programming method and the linear programming method. These two methods are intricately connected, with the former also to be considered as the dual of the latter \cite{HeLa96}.

In the literature, however, a quantum counterpart has not been studied in complete generality despite the emergence of many recent studies in a variety of contexts \cite{BaBaAa14,YiYi18,YiFeYi21,SaYu22} as we review in the following. The first and most well-known work that extends the classical MDP model to the quantum domain is \cite{BaBaAa14}. In this paper, the authors introduce quantum observable Markov decision processes (QOMDPs) as a quantum counterpart to partially observable Markov decision processes (POMDPs). In a QOMDP, the state is represented by a density operator and the set of actions is a finite collection of super-operators (quantum channels) that can be applied to these states. In a QOMDP, a chosen super-operator is decomposed into its Kraus operators, with each Kraus operator representing an observation. Therefore, the next state is determined based on the outcome of the chosen super-operator and the cost is paid depending on the chosen action and the current state. In this work, the authors primarily compare complexity-related problems between classical and quantum settings.  In the future work section of \cite{BaBaAa14}, authors pose the following problem:

\begin{quote}
"We also proved computability results, but did not consider algorithms for solving any of the problems we posed beyond a very simple PSPACE algorithm for policy existence. Are there quantum analogues of POMDP algorithms or even MDP ones?". 
\end{quote}

\noindent Indeed, the ultimate objective of this work is to construct quantum counterparts of the dynamic programming and the linear programming methods, with the latter to be established in a forthcoming paper. These quantum analogues serve as tools for calculating optimal policies and optimal value functions of q-MDPs for various classes of policies. Therefore, we provide a comprehensive solution to the aforementioned problem, albeit employing a more general quantum MDP model.

Another closely related paper is \cite{YiYi18}, where the authors introduce another quantum version of the MDPs. In this version, in addition to measurement super-operators as in QOMDPs, authors also utilize indivisible super-operators to model the state evolution. Namely, in this model, the state evolves with respect to chosen indivisible super-operators, but at times, the controller is allowed to take measurements on the system state via divisible super-operators. The objective of this paper is also centered around some complexity-related problems.

Another paper that addresses the quantum version of MDPs is \cite{YiFeYi21}. This paper is indeed the most closely related to ours in terms of the objective, which is computing optimal policies and value functions. In their model, the state evolves as follows: Given the current state, represented as a density operator on some finite-dimensional Hilbert space, an action is chosen from a finite set. Then, the state transitions to an intermediate state through a super-operator parameterized by the set of actions. Finally, a measurement is conducted on the intermediate state to obtain the reward function as well as the next state. 

Since the models mentioned above share similarities, we refer to them collectively as QOMDPs, using this as an umbrella term. It is worth noting that QOMDP and q-MDP models exhibit some distinctions. Firstly, one can formulate the QOMDP problem as a classical MDP with an uncountable compact state space (i.e. the set of density operators on some finite dimensional Hilbert space) and finite action space. Consequently, classical methodologies like dynamic programming and linear programming could be applicable, perhaps with some discrete approximation of the state space. Unfortunately, in q-MDP context, direct application of classical techniques is not feasible; they need to be adapted to suit our setup. This is true especially for the linear programming formulation.

Additionally, since one can write QOMDP as a classical MDP with uncountable compact state space, one could further transform this classical MDP into a deterministic one by elevating the state and action spaces to the set of probability measures on the state and state-action spaces. Then, this formulation could be further transformed into q-MDP since classical deterministic MDPs over probability distributions can always be represented as a special case of our quantum model. Hence, we can view the QOMDPs as a specialized instance of q-MDPs. However, a challenge arises here: the corresponding Hilbert space for the state must be an infinite-dimensional Hilbert space, a scope not covered in this paper. However, we can handle this complexity by approximating the classical version of QOMDP with finite state MDPs, whose convergence can be established to the underlying MDP via the results in \cite{SaYuLi17}. We can then embed the finite-state approximate MDPs into q-MDP formulation with finite-dimensional Hilbert spaces, enabling us to approximately model QOMDP as a special case of the q-MDP with finite dimensional Hilbert spaces. 

In \cite{SaYu22}, the quantum generalization of static decentralized control systems is discussed. Since the model is static, it lacks state dynamics and therefore lacks dynamical difference equations, unlike the MDP problem. In this generalization, the primary objective is to induce additional correlations among decentralized decision-makers by allowing them to share an entangled state and perform independent quantum measurements on their respective parts of the entangled state. This approach enables the creation of extra correlations that can improve system performance. We can consider this generalization of static decentralized control systems as semi-quantum, as everything in the model remains unchanged except for the conditional distribution of each agent's action given its observation, which is determined by the outcome of the quantum measurement. Nonetheless, our paper's quantum MDP framework can be generalized to such multi-agent and decentralized dynamic quantum formulations by unifying the model presented in our paper and the space of decentralized control policies studied in \cite{SaYu22}.

There have also been numerous fundamental contributions in the control theory literature concerning the optimal control of quantum systems. These works predominantly focus on the continuous-time setting \cite{AlTi12,Dal08,DoPe22}. As evident from classical optimal control theory, despite the conceptual similarities between approaches developed for continuous-time and discrete-time systems, there exist significant technical distinctions. For instance, in the continuous-time scenario, optimal control policies are often derived by solving a partial differential equation -- the Hamilton-Jacobi-Bellman (HJB) equation -- associated with the system. To ensure the existence of well-behaved solutions to this equation, stringent regularity assumptions must be imposed on the system components. Conversely, in discrete-time, optimal control policies are typically obtained through dynamic programming or other discrete-time optimization methods like linear programming. In this context, there is no need to impose stringent regularity conditions on the system components; generally, mere continuity of the system components suffices. Consequently, one might anticipate a similar discrepancy in the quantum domain.

Furthermore, there are also conceptual differences between the methodologies employed in control theory literature and our approach to modeling quantum control systems. In conventional control theory literature, a quantum control system is typically considered a direct extension of deterministic control models, expressed through controlled differential equations. In the quantum domain, the governing differential equation for quantum dynamics is the controlled Schr\"{o}dinger equation. 

If the density operator, controlled by the Schr\"{o}dinger equation, only represents the state of the quantum system, it leads to closed quantum control systems, implying no interaction with the environment. This is due to the corresponding semi-group of operators controlling the evolution of the state being unitary operators. When dealing with interactions between the quantum system and the environment (open quantum systems), the Schr\"{o}dinger equation is defined over the composition of the quantum system and the environment. Obtaining the reduced description of this state dynamics over the quantum system of interest can be done via partial trace operation, which in general results in non-Markovian dynamics \cite{AlTi12}. This is similar to classical partially observable MDPs, where the joint state of two subsystems forming a Markovian process does not imply that an individual state, which may be viewed as a measurement of the hidden Markov model, is itself Markovian. Non-Markovianity arises due to the system's interaction with an environment, where memory effects emerges in the reduced system. Although there has been some exploration into controlling non-Markovian dynamical systems, the majority of control-oriented literature has focused on the Markovian scenario due to its significant advantage of yielding simpler dynamical equations \cite{AlTi12}, where numerous methods from classical control theory are available to analyze such dynamical equations. Indeed, under appropriate assumptions, non-Markovian dynamics can be approximated by a Markovian structure, leading to a Markovian master equation in the Lindblad form \cite{Lin76}. For instance, in the weak coupling limit, one can derive the Markovian master equation, where the joint state of the quantum system and the environment is approximated by a product state, with the environment state remaining invariant under the Hamiltonian action \cite{BrPe07}. In our formulation, as we do not begin with a Hamiltonian description of the quantum system, the model lacks the memory effect that typically leads to non-Markovian dynamics.

In both closed quantum systems and Markovian approximations of open quantum systems, the quantum control model can be regarded as a classical control model over the set of density operators. Consequently, much of the work in quantum control focuses on classical control problems such as stability \cite{TiVi09} and controllability \cite{Alt03,Alt04}. Various papers also address optimal control of quantum systems \cite{PeDaRa88,ZhRa98,Dal08}. As the quantum control system is a direct extension of classical controlled models, well-established techniques from the optimal control literature, such as Pontryagin maximum principle \cite{BoSiSu221} or Hamilton-Jacobi-Bellman formalism \cite{GoBeSm05}, can be developed for those quantum systems by mimicking the classical case. However, addressing the optimal control problem for non-Markovian dynamics in open quantum systems poses a considerable challenge \cite{BiGo12}, given that most methods in classical control theory are developed under the Markovian assumption for state dynamics. This indeed constitutes a significant challenge in quantum control literature for open systems in continuous-time.

In contrast to the continuous-time framework, our approach does not directly extend the controlled stochastic difference equation governing the classical system's state to an equation over a set of density operators in discrete-time.  Instead, we initially lift our classical stochastic control system to the set of probability measures over state and state-action spaces. The dynamics of this extended system can be fully modeled using classical channels, avoiding the need for stochastic difference equations. Since density operators and quantum channels are quantum generalizations of probability measures and classical channels, we promptly derive the quantum version of this extended model. Consequently, our approach circumvents non-Markovianity in the dynamics and accommodates both closed and open systems without requiring distinct descriptions and interaction with the unknown environment. 

The objective of this paper is to develop a quantum version of a classical MDPs. To accomplish this, we begin by converting the classical MDPs into deterministic ones by lifting both the state and action spaces to the set of probability measures. Next, we substitute probability measures with density operators and replace the state transition channel with a quantum channel. This process result in a quantum version of the MDPs. After formulating the q-MDP, we establish a verification theorem that demonstrates the sufficiency of Markovian quantum controls and provides the dynamic programming principle. Subsequently, we conduct a comparison between our q-MDP model and  QOMDPs. Furthermore, we derive approximations for q-MDPs through finite-action models, where these approximations can be cast as QOMDPs. Lastly, we introduce classes of policies for q-MDPs such as open-loop policies and classical-state-preserving closed-loop  policies, and obtain the structural results of these policies.

\subsection{Contributions}
\begin{itemize}
\item[\ding{43}] Inspired by the deterministic reduction of classical MDPs, we introduce q-MDPs in Section~\ref{sec2new-2} with the most general setup. Here, the state is a density operator on a Hilbert space whose dimension is typically the cardinality of the state space of the classical MDP, the action is a density operator on a Hilbert space whose dimension is typically the cardinality of the state-action space of the classical MDP, and the state transition channel is a quantum channel from the latter Hilbert space to the former one. In Section~\ref{classical-q}, we demonstrate how d-MDPs can be represented as q-MDPs and introduce the quantum versions of classical policies.
\item[\ding{43}] In Section~\ref{sec2new-3}, we establish the verification theorem for q-MDPs (see Theorem~\ref{verification}). Specifically, we demonstrate that storing the entire past information is not necessary for optimality and establish the existence of an optimal quantum Markov policy. To achieve this, we introduce the dynamic programming operator for q-MDPs and employ it in the proof of the verification theorem.
\item[\ding{43}] In Section~\ref{sec2new-4}, we compare q-MDPs with QOMDPs. We achieve this by transforming any QOMDP into a classical MDP, denoted as c-QOMDP, with an uncountable compact state space, representing the set of density operators on a finite-dimensional Hilbert space. Utilizing established results on the approximation of classical MDPs with uncountable compact state spaces, we then approximate the c-QOMDP with finite-state classical MDPs, denoted as c-QOMDP$_n$ (see Theorem~\ref{compact:mainthm1}). As demonstrated in Section~\ref{classical-q}, any finite-state MDP can be integrated into our q-MDP model, allowing us to present c-QOMDP$_n$ as a specialized case of q-MDP with classical policies. Consequently, we can approximate any QOMDP using q-MDPs equipped with classical policies. This suggests that, at least in an approximate sense, q-MDPs exhibit a broader generality compared to QOMDPs.
\item[\ding{43}] In Section~\ref{sec2new-5}, we introduce a finite-action approximation, denoted as q-MDP$_n$, for q-MDPs, formulating it as a specific case of QOMDP. We then demonstrate its convergence to the original q-MDP  (see Theorem~\ref{compact:mainthm3}). Consequently, using this result, we can initially approximate any q-MDP with q-MDP$_n$ that is a specific QOMDP instance. Subsequently, by recasting q-MDP$_n$ as a classical MDP (given its nature as a QOMDP), we can further approximate it using finite-state MDPs following  Theorem~\ref{compact:mainthm1}.
\item[\ding{43}] In Section~\ref{sec1}, we define open-loop and classical-state-preserving closed-loop  quantum policies, respectively. We introduce the open-loop quantum policies by considering natural extensions of inverse and marginalization operators in the classical case to the quantum domain. The structure of quantum policies given by these specific descriptions of inverse and marginalization operators is established in Proposition~\ref{structure-qmdp}. It becomes apparent that these quantum policies do not incorporate state information but solely rely on the initial density operator. As a result, we categorize these policies as \emph{open-loop}. Then, our focus shifts to the definition of \emph{so-called} classical-state-preserving closed-loop  quantum policies by changing the descriptions of the inverse and marginalization operators. The need for such policies arises because open-loop policies, which do not use the state information, have limitations as they do not encompass the classical setting. Consequently, we present a definition for classical-state-preserving closed-loop  quantum policies, where we relax the constraint in the definition of inverse operator. Upon introducing the definition of the classical-state-preserving closed-loop  policies, we obtain the structure of these classical-state-preserving closed-loop  quantum policies in Proposition~\ref{structure-qwmdp}.
\end{itemize}

Below in Figure~\ref{Hier}, one can find the hierarchy of policies outlined in this paper, providing a structured overview of the different policy classes discussed and their relationships.

\begin{figure}[h]
\centering
\tiny
\begin{tikzpicture}[scale=0.75]

\draw[thick] (-8,-6) rectangle (8,6);

\draw[thick] (0,0) ellipse (5.5 and 5.5);
\node at (-5.5,5.5) {\textbf{Quantum Policies (QPs)}};
\node at (-0,5) {\textbf{Markov QPs}};

\draw[thick] (0,0) ellipse (4.5 and 4.5);
\node at (0,3) {\textbf{Closed-loop QPs}};
\node at (0,3.5) {\textbf{Clasical-State-Preserving}};

\draw[thick] (-1.5,-0.5) ellipse (2.6 and 2.6);
\node at (-2.5,-0.5) {\textbf{Classical QPs}};

\draw[thick] (1.5,-0.5) ellipse (2.6 and 2.6);
\node at (2.5,-0.5) {\textbf{Open-loop QPs}};

\end{tikzpicture}
\normalsize
\caption{Hierarchy of policies: (i) Quantum policies utilize all historical information. (ii) Markov quantum policies rely solely on current state information without any additional constraints. (iii) Classical-state-preserving closed-loop  quantum policies also use the current state information but with a relaxed invertibility condition on the quantum channel, mirroring the classical invertibility condition. (iv) Open-loop quantum policies similarly use current state information but impose a stricter invertibility condition on the quantum channel, again reflecting the classical case. (v) Classical quantum policies represent an embedding of classical policies within a quantum framework.}
\label{Hier}
\end{figure}
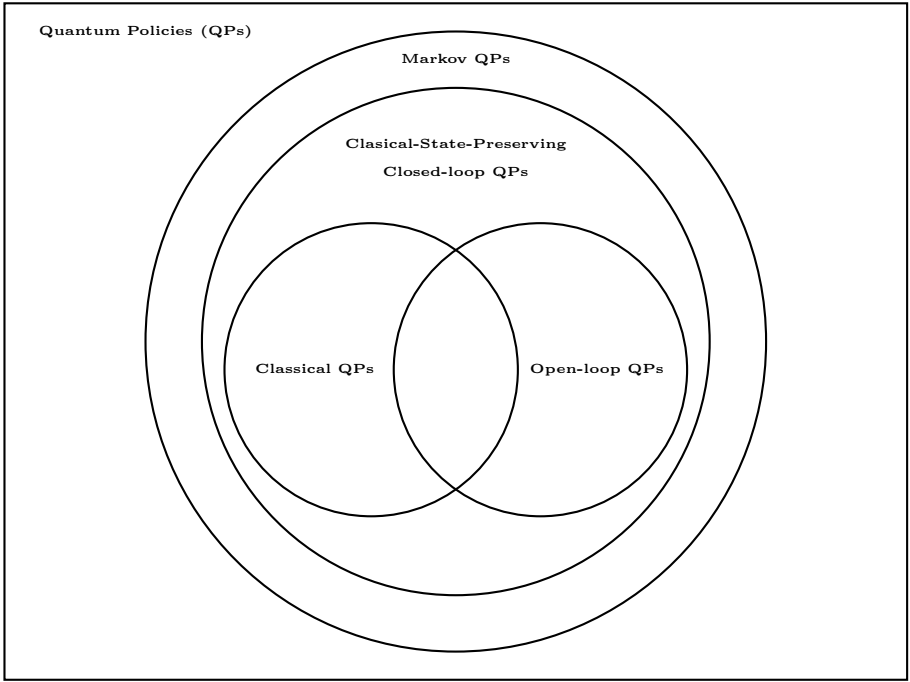

\section{Quantum Markov Decision Processes}\label{sec2new}

In this section, our goal is to introduce q-MDPs, but before diving into this topic, we first aim to establish the motivation behind this model. To do so, we start by laying out the formulation of classical MDPs and their deterministic reduction, a well-documented procedure that involves expanding the state and action spaces to the set of probability measures on the state and the state-action spaces. This reduction, which can be found in Section 9.2 of \cite{BeSh78}, serves as a crucial foundation for our formulation. Building on this deterministic MDP framework, we then proceed to introduce q-MDPs.

\subsection{Classical Markov Decision Processes}\label{sec2new-1}
A discrete-time classical MDP can be described by a five-tuple $\bigl( \sX, \sA, p, c \bigr)$, where Borel spaces $\sX$ and $\sA$ denote the \emph{state} and \emph{action} spaces, respectively. Here, we assume that $\sA$ is compact. The \emph{stochastic kernel} $p(\,\cdot\,|x,a)$ denotes the \emph{transition probability} of the next state given that previous state-action pair is $(x,a)$. The \emph{one-stage cost} function $c$ is a Borel measurable function from $\sX \times \sA$ to $\R$. A policy\footnote{We note that, in general, a policy at time $t$ can depend on the past history $h(t) \coloneqq (x_0,a_0,\ldots,x_t)$. However, due to the Markovian state dynamics, we can, without loss of generality, confine the search space to policies that solely rely on the current state information. This is why the policy is defined as above, utilizing only the current state information.} is a sequence $\bpi=\{\pi_{t}\}$ of stochastic kernels from $\sX$ to $\sA$. The set of all policies is denoted by $\Pi$. Let $\Phi$ denote the set of stochastic kernels $\varphi$ from $\sX$ to $\sA$. A \emph{stationary} policy is a constant sequence $\bpi=\{\pi_{t}\}$ where $\pi_{t}=\pi$ for all $t$ for some $\pi \in \Phi$.  The set of stationary policies is identified with the set $\Phi$. Under an initial distribution $\mu_0$ on $\sX$ and a policy $\bpi$, the evolutions of the states and actions are given by
\begin{align}
x_0 \sim \mu_0, \nonumber \,\,\,\,
x_t \sim p(\,\cdot\,|x_{t-1},a_{t-1}), \text{ } t\geq 1, \,\,
a_t \sim \pi_t(\,\cdot\,|x_t), \text{ } t\geq0. \nonumber
\end{align}
In MDPs, a typical goal is to minimize the cost structures such as finite horizon cost, discounted cost, and average (or ergodic) cost criteria, denoted as $J(\bpi,\mu_0)$. The optimal value of the MDP is then defined by $J^*(\mu_0) \coloneqq \inf_{\bpi \in \Pi} J(\bpi,\mu_0)$. A policy $\bpi^{*}$ is said to be optimal for $\mu_0$ if $J(\bpi^{*},\mu_0) = J^*(\mu_0)$. 

By using a process similar to the one described in Section 9.2 of \cite{BeSh78}, we can convert any MDP into an equivalent deterministic model.

To introduce the equivalent MDP, for each $\mu \in \P(\sX)$, we define 
$\sA(\mu) \coloneqq \left\{\nu \in \P(\sX \times \sA): v(\,\cdot\,\times\sA) = \mu(\,\cdot\,)\right\}.$ 
Then, the deterministic reduction of a MDP, denoted as d-MDP, can be described by a tuple 
$
\left( \P(\sX), \P(\sX \times \sA),  P, C \right),
$
where $\P(\sX)$ is the state space and $\P(\sX \times \sA)$ is the action space, each endowed with weak convergence topology. The state transition function $P: \P(\sX \times \sA) \rightarrow \P(\sX)$ is given by
\begin{align}
P(\nu)(\,\cdot\,) = \int_{\sX \times \sA} p(\,\cdot\,|x,a) \, \nu(dx,da). \label{eq1}
\end{align}
The \emph{one-stage cost} function $C$ is a linear function from $\P(\sX \times \sA)$ to $\R$ and is given by
\begin{align}
C(\nu) = \int_{\sX \times \sA} c(x,a) \, \nu(dx,da) \eqqcolon \langle c,\nu \rangle. \label{eq2}
\end{align}
In this model, a Markov policy is a sequence of classical channels $\bgamma=\{\gamma_{t}\}$ from $\sX$ to $\sX\times\sA$ where for all $t$
\begin{align} 
\gamma_t(\mu) \in \sA(\mu), \,\, \forall \,\mu \in \P(\sX). \label{inv_classical}
\end{align}
The set of all policies is denoted by $\Gamma$. A \emph{stationary} policy is a constant sequence of channels $\bgamma=\{\gamma\}$ from $\sX$ to $\sX\times\sA$, where $\gamma(\mu) \in \sA(\mu)$ for all $\mu \in \P(\sX)$.

\subsection{Quantum Systems}

Before delving into the formulation of q-MDPs, it is helpful to review some foundational concepts in quantum mechanics and establish a connection with classical systems. For computational purposes, this paper primarily focuses on finite-dimensional quantum systems. Accordingly, in this section, we provide all relevant definitions within the context of these quantum models. We refer the reader to the books \cite{NiCh02,Wil13,Wat18,Par92} for the terminology and definitions used in this section.

For any quantum system, there is a finite-dimensional complex Hilbert space $\H$, which is known as the state space of the system. Following the convention in quantum mechanics, we use the Dirac notation to denote the elements of $\H$ as a ket-vector $|\psi \rangle$. Then, the conjugate transpose of any ket-vector $|\psi \rangle$ is denoted by a bra-vector $\langle \psi |$. With this convention, the inner product in $\H$ can be written as $\langle \psi | \xi \rangle$. 

According to the first postulate of quantum mechanics, a state that completely describes a quantum system is a density operator $\rho$, which is a positive semi-definite operator with unit trace, acting on the state space $\H$, i.e.,
$
\langle \psi | \rho |  \psi \rangle \geq 0 \,\, \forall \psi \in \H \,\, \text{and} \,\, \tr(\rho) \coloneqq  \sum_{i=1}^n \rho_{ii} = 1,
$
where $n$ is the dimension of $\H$ and $\{\rho_{ii}\}$ are the diagonal elements of $\rho$ in any orthonormal basis of $\H$. If $\rho$ is a rank-one operator; i.e., $\rho = |\psi \rangle \langle \psi|$, where $|\psi \rangle$ is a unit vector, then $\rho$ is called a pure state. Otherwise, it is called a mixed state. Let $\D(\H)$ denote the set of all density operators on $\H$. It is known that $\D(\H)$ is a convex set and its extreme points are the set of pure states \cite[page 46]{Wat18}. 

A canonical example for a density operator is a probability distribution $\mu$ on the set $[n] \coloneqq \{1,2,\ldots,n\}$, where $n$ is the dimension of $\H$. Specifically, if we represent any density operator as a matrix with respect to a fixed orthonormal basis, then defining $\rho \coloneqq \diag\{\mu(i)\}$ results in a density operator, as $\langle \psi | \rho | \psi \rangle \geq 0 \,\, \forall \psi \in \H$ because $\mu(i) \geq 0$ for all $i$, and $\tr(\rho) = 1$ since $\sum_{i=1}^n \mu(i) = 1$. Thus, any probability distribution on $[n]$ can be viewed as a density operator on $\H$. These density operators are known as classical states. Observe that being classical apriori requires that a basis is fixed. Quantum states are thus strictly larger than classical states.

Following the analogy above, let us explore how to perform the expectation operation in the quantum domain. Consider \( X \), a random variable defined on the probability space \(([n], \mu)\). In the quantum framework, we can represent this random variable as \(\phi \coloneqq \text{diag}\{X(i)\}\). In view of this, the expectation of \( X \) with respect to the probability distribution \(\mu\) can then be calculated in the quantum domain as:
$
\cE[X] = \langle X, \mu \rangle \coloneqq \sum_{i} X(i) \mu(i) = \tr(\rho \phi) \coloneqq \langle \rho, \phi \rangle,
$
where the first inner product is the Euclidean inner product, and the second is the Hilbert-Schmidt inner product. Thus, in the quantum domain, the expectation operation is carried out using the trace operation.

According to the second postulate of the quantum mechanics, the evolution of an isolated quantum systems, named as closed quantum system, is described by a unitary transformation. That is, the next state $\rho^{+}$ of a closed quantum system is related to the current state $\rho$ via a unitary operator $U$ (i.e., $U U^* = U^* U = \Id$): 
$
\rho^+ = U \rho U^*. 
$
For instance, the bit-flip operator is a unitary evolution of a two-dimensional quantum system:
$$
\begin{pmatrix}
0 & 1 \\ 1 & 0
\end{pmatrix}
$$
It sends the q-bit $|0\rangle$ representing $[0 \,\, 1]$ to the q-bit $|1\rangle$ representing $[1 \,\, 0]$ and the q-bit $|1\rangle$ to the q-bit $|0\rangle$. 

In general, it is challenging to prevent the interaction of quantum system with its environment in order to achieve a closed quantum system. When such interactions occur, then this class of quantum systems are referred to as open quantum systems. Their evolution can be described as follows. First of all, if we have two quantum systems with the corresponding Hilbert spaces $\H_1$ and $\H_2$, the combined system is described by the tensor product Hilbert space $\H_1 \otimes \H_2$. With this in mind, let $\H_e$ represent the Hilbert space associated with the environment system, and let $\rho_e$ be the current state of the environment. The combined state of the system and environment is then given by the product $\rho \otimes \rho_e$. Since we have included all possible interactions within the environment, we can treat the combined system $\H \otimes \H_e$ as a closed quantum system. Therefore, the evolution of this combined state can be described by a unitary operator $U_c$ as follows:
$
\xi^+ = U_c (\rho \otimes \rho_e) U_c^*, 
$
where $\xi^+ \in \D(\H \otimes \H_e)$. Note that $\xi^+$ may not remain in a simple product form. To recover the next state of our quantum system $\H$ from this combined state $\xi^+$, we use the partial trace operation.

The partial trace operation, denoted by $\tr_{E}: \D(\H \otimes \H_e) \rightarrow \D(\H)$, is defined as $\tr_E \coloneqq \Id \otimes \tr$. Specifically, when applied to product states like $\rho \otimes \rho_e$, it acts as:
\[
\tr_E(\rho \otimes \rho_e) = \rho \, \tr(\rho_e) = \rho \quad \text{since} \quad \tr(\rho_e) = 1.
\]
The partial trace operation can be seen as a quantum analogue of marginalization. For instance, if $\nu$ is a probability distribution on $[n] \times [m]$, where $m$ is the dimension of $\H_e$, the marginalization on $[m]$ is given by:
\[
\Mar_{E}(\nu) = \nu(\cdot \times [m]) \eqqcolon \mu(\cdot),
\]
where $\mu$ is the marginal distribution of $\nu$ on $[n]$. If $\xi$ and $\rho$ are density operators corresponding to the distributions $\nu$ and $\mu$, respectively, then one can show that
\[
\rho = \tr_E(\xi).
\]
This means that the marginalization operation $\Mar_{E}$ in the classical domain is implemented in the quantum domain via the partial trace operation $\tr_E$.

Returning to our original question of how to recover the next state of our quantum system from the next state of the combined system, we now see that this is done through the partial trace operation:
\[
\rho^+ = \tr_E(U_c (\rho \otimes \rho_e) U_c^*) \eqqcolon \N(\rho).
\]
In other words, we trace out the environment's state from the combined state to obtain the next quantum state of our system. This operation can be viewed as a linear map on density operators, denoted by $\rho \mapsto \N(\rho)$. In the literature, since these operators act on other operators, they are referred to as super-operators.

These super-operators are not only linear but also satisfy the following properties:
\begin{itemize}
\item[\ding{43}] \textbf{Trace preservation}: For any linear operator $\rho$ on $\H$ (not necessarily a density operator, so it may not have unit trace), the super-operator $\N$ preserves the trace, meaning $\tr(\N(\rho)) = \tr(\rho)$. This is because such super-operators must map density operators to density operators.
\item[\ding{43}] \textbf{Complete positivity}: For any quantum environment system $\H_e$, the super-operator $\N \otimes \Id$ acting on $\D(\H \otimes \H_e)$ is positive, meaning it maps positive semi-definite operators to positive semi-definite operators. This property is crucial because the positivity of $\N$ alone is not enough; the combined operator $\N \otimes \Id$ must also be positive for any environment quantum system $\H_e$.
\end{itemize}
It is worth noting that the converse is also true: if $\N$ is trace-preserving and completely positive, then there exists an environment quantum system $\H_e$ such that $\N$ can be realized in this way (see \cite[Section 8.2]{NiCh02}).

In quantum information theory, super-operators that are both trace-preserving and completely positive are known as quantum channels. Quantum channels can also connect quantum systems of different dimensions, serving as quantum analogs of classical channels, as we will discuss shortly. Before diving into that, let us explore how to represent quantum channels explicitly. There are several ways to represent quantum channels, such as the Choi representation, Kraus representation, and Stinespring representation. In this paper, we focus on the Kraus representation (also known as operator-sum representation).

Let \( \N: \D(\H_1) \rightarrow \D(\H_2) \) be a super-operator. This super-operator is a quantum channel, meaning it is both trace-preserving and completely positive, if and only if there exist operators \( \{K_l\} \) from \( \H_1 \) to \( \H_2 \) such that (see \cite[Chapter 4]{Wil13})
$$
\N(\rho) = \sum_l K_l \rho K_l^* \,\, \text{and} \,\, \sum_l K_l^* K_l = \Id.
$$

As mentioned earlier, quantum channels are simply quantum generalizations of classical channels. Therefore, a fundamental example of a quantum channel is, in fact, a classical channel. Consider a classical channel \( W: [n] \rightarrow \P([m]) \) that maps each input \( i \in [n] \) to a probability distribution \( W(\cdot|i) \) on the set \([m]\). We can also represent this map as \( W: \P([n]) \rightarrow \P([m]) \), where
\[
W(\mu)(\cdot) \coloneqq \sum_i W(\cdot|i) \, \mu(i).
\]
Here, for simplicity, we use the same notation \( W \) for both mappings.

If we consider probability distributions as classical states, then this \( W \) mapping corresponds to a super-operator. In fact, it can be viewed as a quantum channel with the form:
\[
\N_W(\rho) = \sum_{\substack{j \in [m] \\ i \in [n]}} \sqrt{W(j|i)} \, | j \rangle \, \langle i| \, \rho \, |i\rangle \, \langle j | \, \sqrt{W(j|i)}.
\]
Here, \( \{|i\rangle\} \) is an orthonormal basis for \( \H_{[n]} \) and \( \{|j\rangle\} \) is an orthonormal basis for \( \H_{[m]} \), where the dimensions of \( \H_{[n]} \) and \( \H_{[m]} \) are \( n \) and \( m \), respectively. The Kraus operators associated with \( \N_W \) are \( \{\sqrt{W(j|i)} \, | j \rangle \, \langle i|\}_{i,j} \), and they satisfy the condition:
\[
\sum_{i,j}  (\sqrt{W(j|i)} \, | j \rangle \, \langle i|)^* \, \sqrt{W(j|i)} \, | j \rangle \, \langle i| = \sum_{i,j} |i\rangle \, \langle j | \, \sqrt{W(j|i)} \, \sqrt{W(j|i)} \, | j \rangle \, \langle i| = \Id.
\]
It can be shown that the quantum version of the probability distribution \( W(\mu) \) is exactly the same as \( \N_W(\rho) \), where \( \rho \) is the quantum analog of \( \mu \). Therefore, any classical channel can also be realized in the quantum domain, illustrating that quantum channels are indeed generalizations of classical channels. 

Before we dive into the formulation of q-MDPs, let us clarify one final point. Consider a bijective mapping \( f: [n] \rightarrow [n] \), which is essentially a permutation of the finite set \([n]\). This mapping defines a classical channel \( W(\cdot|i) \coloneqq \delta_{f(i)}(\cdot) \). The quantum equivalent of this classical channel can be expressed as 
\[
\N_W(\rho) = U \rho U^*,
\]
where \( U \) is the permutation matrix corresponding to the function \( f \) in matrix form.  Obviously, $U$ is unitary transformation. Hence, if the classical channel $W$ is realized by a bijective function, then in the quantum domain, this corresponds to the closed quantum system. In this sense, open quantum systems are generalizations of stochastic evolutions in classical setups, while closed quantum systems correspond to (bijective) deterministic evolutions in the classical setup.

With this understanding, we are now ready to introduce q-MDPs.

\subsection{Quantum Markov Decision Processes: A General Formulation}\label{sec2new-2}

Inspired by the connections between classical systems and quantum systems established in the previous section, and utilizing the deterministic reduction of MDPs previously discussed, we are now prepared to formulate the quantum Markov Decision Processes (q-MDPs). The development of this quantum model is guided by the following key observations:
\begin{itemize}
\item[\ding{43}] In quantum information and computation, traditional probability distributions are extended and substituted with density operators, as discussed earlier.
\item[\ding{43}] Classical channels, which are used to describe the probabilistic evolution of classical systems, are generalized and replaced with quantum channels. These quantum channels account for the more complex interactions and state changes that can occur in a quantum system.
\item[\ding{43}] The expectation operator, which is used to calculate expected values in classical systems, is generalized and replaced by the trace operator.
\end{itemize}
With these foundational concepts and analogies in place, we can now proceed to formally define the q-MDPs. This new framework will allow us to analyze and design quantum decision-making processes, extending the capabilities of classical MDPs into the quantum realm.

Initially, we assume that all Hilbert spaces in the description of q-MDPs are separable and possibly infinite dimensional. Later we confine ourselves to the finite dimensional setting. To this end, let $\H_{\sX}$ and $\H_{\sA}$ be separable complex Hilbert spaces that may have infinite dimensions. 

A discrete-time quantum Markov decision process (q-MDP) can be described by a tuple
\begin{align}
\bigl( \D(\H_{\sX}), \D(\H_{\sX} \otimes \H_{\sA}), \N, C \bigr), \nonumber
\end{align}
where  $\D(\H_{\sX})$ is the state space, $\D(\H_{\sX} \otimes \H_{\sA})$ is the action space, and the state transition super-operator 
$$\N: \D(\H_{\sX} \otimes \H_{\sA}) \rightarrow \D(\H_{\sX})$$ 
is a quantum channel. Hence, given the current action in $\D(H_{\sX} \otimes \H_{\sA})$, the next state is the output of the quantum channel $\N$; that is, the state dynamics is of the Markov type, in the sense that, the past is irrelevant to compute the next quantum state. 

The one-stage cost function $C$ is a linear function from $\D(\H_{\sX} \otimes \H_{\sA})$ to $\R$ of the following form
$$
C(\sigma) = \sTr(c \sigma) \eqqcolon \langle c,\sigma \rangle
$$
for some Hermitian operator $c$ from $\H_{\sX} \otimes \H_{\sA}$ to itself. Here, $\sTr(c \sigma) \eqqcolon \langle c,\sigma \rangle$ is Hilbert-Schmidt inner product. 

An admissible quantum policy constitutes a sequence of quantum channels $\bgamma=\{\gamma_{t}\}$ such that 
$$
\gamma_t: \left(\bigotimes_{k=0}^t \D\left(\H_{\sX}\right) \right) \otimes \left(\bigotimes_{k=0}^{t-1}\D\left(\H_{\sX} \otimes \H_{\sA}\right)\right) \rightarrow \D(\H_{\sX} \otimes \H_{\sA}),
$$ that is, $\gamma_t$ maps the history $\rho_0 \otimes \sigma_0 \otimes \ldots \otimes \rho_{t-1} \otimes \sigma_{t-1} \otimes \rho_t$ to an action $\sigma_t$ at time $t$. The set of all admissible policies is denoted by $\Gamma$. An admissible quantum policy is said to be Markov if it is a sequence of quantum channels $\bgamma=\{\gamma_{t}\}$ from $\H_{\sX}$ to $\H_{\sX} \otimes \H_{\sA}$; that is, $\gamma_t$ only uses the current state information $\rho_t$ at each time $t$. The set of all Markov policies is denoted by $\Gamma_m$. A Markov policy $\gamma$ is called stationary if $\gamma_t = \gamma_{t'}$ for all $t,t'$. 

Here, the evolutions of the states and the actions are given by
\begin{align}
\rho_0 &= \rho_0, \,\,\,\,
\rho_{t+1} = \N(\sigma_t), \text{ } t \geq 0 \nonumber \\
\sigma_t &= \gamma_t(\rho_0 \otimes \sigma_0 \otimes \ldots \otimes \rho_{t-1} \otimes \sigma_{t-1} \otimes \rho_t), \text{ } t \geq 0. \nonumber
\end{align}
This model can accommodate different cost structures as in classical MDPs, including finite horizon cost, discounted cost, and average cost criteria, which are defined as follows:
\begin{align*}
V(\bgamma,\rho_0) &= \sum_{t=0}^{T-1} C(\sigma_t) + C_T(\sigma_T) = \sum_{t=0}^{T-1} \langle c,\sigma_t \rangle + \langle c_T,\sigma_T \rangle \,\, \text{(finite horizon cost)}\\
V(\bgamma,\rho_0) &= \sum_{t=0}^{\infty}\beta^{t} C(\sigma_t) =  \sum_{t=0}^{\infty}\beta^{t} \langle c,\sigma_t \rangle \,\, \text{(discounted cost)}\\
V(\bgamma,\rho_0) &= \limsup_{T \rightarrow \infty} \frac{1}{T}\sum_{t=0}^{T-1} C(\sigma_t) =  \limsup_{T \rightarrow \infty}\frac{1}{T}\sum_{t=0}^{T-1} \langle c,\sigma_t \rangle  \,\, \text{(average cost)}
\end{align*}
While it is possible to extend the analysis developed for q-MDPs to cover finite horizon cost and average cost criteria, this work concentrates primarily on discounted cost. The examination of finite horizon cost and average cost criteria is considered as a future research direction.

The discounted optimal value of the q-MDP is defined as
\begin{align}
V^*(\rho_0) &\coloneqq \inf_{\bgamma \in \Gamma} V(\bgamma,\rho_0). \nonumber
\end{align}
A policy $\bgamma^{*}$ is said to be optimal for $\rho_0$ if $V(\bgamma^{*},\rho_0) = V^*(\rho_0)$.

\begin{example}

Many algorithms in quantum computing need the initial state to be set to a particular pure state \cite{Gri09}. In this example, we formulate this problem, named as state preparation problem, as a q-MDP. To this end, let $|\psi_{\Target}\rangle \langle \psi_{\Target}|$ be the target pure state in our state space $\H_{\sX}$ that we want to reach at some finite horizon. Let $$\N: \D(\H_{\sX} \otimes \H_{\sA}) \rightarrow \D(\H_{\sX})$$ be our state transition super-operator. Depending on the quantum physical system in question, this super-operator can have various forms. For simplicity, let us assume that \(\mathcal{N}\) is a quantum-to-classical channel, meaning it converts any quantum state into a classical state through a measurement process. 

To describe this operation mathematically, let $\{|x\rangle\}$ be an orthonormal basis for $\H_{\sX}$ and let $\{\Lambda_x\}$ be a collection of positive semi-definite operators on $\H_{\sX} \otimes \H_{\sA}$ satisfying $\sum_x \Lambda_x = \Id$\footnote{According to the third postulate of quantum mechanics, the measurement of any quantum system \(\mathcal{H}\) is described using a positive operator-valued measure (POVM) \(\{\Lambda_i\}\). A POVM is a collection of positive semi-definite operators that sum up to the identity operator, \(\mathbb{I}\), such that: $\sum_i \Lambda_i = \mathbb{I}.$ When a quantum system is in the state \(\rho\), the probability of observing a particular outcome \(i\) during a measurement is determined by the Born rule. This rule states that the probability is given by the trace of the product of the state \(\rho\) and the corresponding operator \(\Lambda_i\), mathematically expressed as $\tr(\rho \Lambda_i).$ This framework allows for the representation of measurements as a process that not only extracts information from the system but also potentially alters its state.}. Then, $\N$ can be written as \cite[Section 4.6.6]{Wil13}
$$
\N(\sigma) \coloneqq \sum_x \tr(\Lambda_x \sigma) \, |x\rangle \langle x|. 
$$
In this scenario, a measurement \(\{\Lambda_x\}\) is performed on \(\sigma\), with the outcome probabilities given by \(\{\tr(\Lambda_x \sigma)\}\). If the measurement yields the outcome \(x\), the state transitions to \(|x\rangle \langle x|\). Since the outcome of the measurement is not declared, the new state is a weighted combination of the pure states \(\{|x\rangle \langle x|\}\), where the weights correspond to the probabilities of obtaining each outcome.

To quantify how close the state \(\rho\) to the target pure state \(|\psi_{\Target}\rangle \langle \psi_{\Target}|\), we can use fidelity, as described in \cite[Section 9.2]{Wil13}:
\[
\mathcal{F}(\rho,|\psi_{\Target}\rangle \langle \psi_{\Target}|) = \tr(\rho |\psi_{\Target}\rangle \langle \psi_{\Target}|) = \langle \psi_{\Target}| \rho |\psi_{\Target}\rangle.
\]
Fidelity \(\mathcal{F}\) ranges from \(0\) to \(1\), serving as a measure of the overlap between two states: a fidelity of \(0\) indicates orthogonal states, while a fidelity of \(1\) signifies identical states via Cauchy-Schwarz inequality. Using fidelity, we define our one-stage cost as follows:
\[
C(\sigma) = 1 - \tr \left( |\psi_{\Target}\rangle \langle \psi_{\Target}| \otimes \mathbb{I}_{\sA} \, \sigma \right) = 1 - \langle |\psi_{\Target}\rangle \langle \psi_{\Target}| \otimes \mathbb{I}_{\sA}, \sigma \rangle \eqqcolon \langle c, \sigma \rangle,
\]
where 
$$c \coloneqq \Id_{\sX \times \sA} - |\psi_{\Target}\rangle \langle \psi_{\Target}| \otimes \mathbb{I}_{\sA}.$$
In the negative part of the cost function \(\tr \left( |\psi_{\Target}\rangle \langle \psi_{\Target}| \otimes \mathbb{I} \, \sigma \right)\), we focus on the overlap between the target state \(|\psi_{\Target}\rangle \langle \psi_{\Target}|\) and the \(\mathcal{H}_{\sX}\)-marginal \(\tr_{\sA}(\sigma)\) of \(\sigma\), which can also be expressed as:
\small
\[
\tr \left( |\psi_{\Target}\rangle \langle \psi_{\Target}| \otimes \mathbb{I}_{\sA} \, \sigma \right) = \tr \left( |\psi_{\Target}\rangle \langle \psi_{\Target}| \, \tr_{\sA}(\sigma) \right) = \F(|\psi_{\Target}\rangle \langle \psi_{\Target}|,\tr_{\sA}(\sigma)).
\]
\normalsize
Thus, the cost function $C$ is an appropriate choice for a one-stage cost function as it effectively penalizes the lack of overlap between the state of our system and the target state.

In this model, our goal is to reach the target state  \(|\psi_{\Target}\rangle \langle \psi_{\Target}|\) via manipulating our state using the policies in an optimal way. To achieve this optimization, we can employ either the finite horizon cost or the discounted cost.
\end{example}

For computational purposes, throughout the remainder of this work, we confine our exploration solely to q-MDPs with finite dimensional Hilbert spaces $\H_{\sX}$ and $\H_{\sA}$. Additionally, we limit our focus to classical MDPs with a finite number of states and actions (i.e., $|\sX| < \infty$ and $|\sA| < \infty$). Consequently, we typically let $\dim(\H_{\sX}) = |\sX|$ and $\dim(\H_{\sA}) = |\sA|$. 

In the following section, we explain how d-MDPs can be represented as q-MDPs and describe the quantum versions of classical policies. 

\subsection{Classical Policies and Quantum Versions of d-MDPs}\label{classical-q}

In this section, we show how to reduce the d-MDPs into q-MDPs. To this end, consider a d-MDP with the components $\bigl( \P(\sX), \P(\sX\times\sA), P, C \bigr)$. Recall that in d-MDPs, a (Markov) policy is a sequence of classical channels $\bgamma=\{\gamma_{t}\}$ from $\sX$ to $\sX\times\sA$ such that for all $t$ 
$$\gamma_t(\mu) \in \sA(\mu), \,\, \forall \,\mu \in \P(\sX).$$
Note that given any classical channel $\gamma_t:\P(\sX) \rightarrow \P(\sX\times\sA)$, there is a corresponding stochastic kernel $W_t$ from $\sX$ to $\sX\times\sA$ where $\gamma_t(\mu)(\, \cdot\,,\,\cdot\,) = \sum_{x \in \sX}W_t(\,\cdot\,,\,\cdot\,|x) \, \mu(x)$. 
If $\gamma_t(\mu) \in \sA(\mu)$ for all $\mu \in \P(\sX)$, then for any $x \in \sX$, we have 
$$
\gamma_t(\delta_x)(\,\cdot\,\times\sA) = \sum_{a \in \sA} W_t(\,\cdot\,,a|x)= \delta_x(\,\cdot\,).
$$
This implies that, for any $x \in \sX$, one can disintegrate $W_t(\,\cdot\,,\,\cdot\,|x)$ as follows
$$
W_t(y,a|x) = \delta_x(y) \, \pi_t(a|x)
$$
for some $\pi_t \in \Phi$. That is, for any  $\gamma_t:\P(\sX) \rightarrow \P(\sX\times\sA)$ satisfying $\gamma_t(\mu) \in \sA(\mu)$ for all $\mu \in \P(\sX)$, there exists a stochastic kernel $\pi_t$ from $\sX$ to $\sA$ such that $\gamma_t(\mu) = \mu \otimes \pi_t$, where $\mu \otimes \pi_t(x,a) \coloneqq \mu(x) \, \pi_t(a|x)$ (this is not a tensor product). Therefore, there is a bijective relation between the set of policies $\Pi$ of the original MDP and the set of policies $\Gamma$ of the d-MDP. This bijective relation indeed establishes the equivalence of d-MDP and the original MDP. 

In the quantum version of the d-MDP, we let $\H_{\sX} \coloneqq \rC^{\sX}$ and $\H_{\sA} \coloneqq \rC^{\sA}$. The quantum representation of the d-MDP has the following components:
\begin{align}
\bigl( \D(\H_{\sX}), \D(\H_{\sX} \otimes \H_{\sA}), \N, C \bigr), \nonumber
\end{align}
where the super-operator $\N: \D(\H_{\sX} \otimes \H_{\sA}) \rightarrow \D(\H_{\sX})$ is defined as 
$$
\N(\sigma) = \sum_{\substack{(x,a) \in \sX\times\sA \\ y \in \sX}} \sqrt{p(y|x,a)} \, | y \rangle \, \langle x,a| \, \sigma \, |x,a\rangle \, \langle y | \, \sqrt{p(y|x,a)}.
$$
Here, \(\{|x\rangle\}\) represents a fixed orthonormal basis for \(\H_{\sX}\), while \(\{|x,a\rangle\}\) is an orthonormal basis for \(\H_{\sX} \otimes \H_{\sA}\), constructed from the orthonormal basis \(\{|x\rangle\}\) of \(\H_{\sX}\) and a fixed orthonormal basis \(\{|a\rangle\}\) for \(\H_{\sA}\). The one-stage cost function $C$ is a linear function from $\D(\H_{\sX} \otimes \H_{\sA})$ to $\R$ of the following form
$$
C(\sigma) = \sTr(c \sigma) \eqqcolon \langle c,\sigma \rangle
$$
with a diagonal operator $c = \diag\{c(x,a)\}_{(x,a) \in \sX\times\sA}$ from $\H_{\sX} \otimes \H_{\sA}$ to itself. In the quantum version of d-MDP, an admissible classical policy is a sequence quantum channels $\bgamma=\{\gamma_{t}\}$ from $\H_{\sX}$ to $\H_{\sX} \otimes \H_{\sA}$, where, for all $t$, there exists a classical channel $\gamma_t: \P(\sX) \rightarrow \P(\sX\times\sA)$ realized by the stochastic kernel $\pi_t$ from $\sX$ to $\sA$ such that 
$$
\gamma_t(\rho) = \sum_{(x,a) \in \sX \times \sA} \sqrt{\pi_t(a|x)} \, |x,a\rangle \, \langle x| \, \rho |x\rangle \, \langle x,a| \, \sqrt{\pi_t(a|x)}. 
$$
Here, using a slight abuse of notation, we represent the quantum counterpart of $\gamma_t$ with the same symbol. Hence, we can embed a d-MDP into a q-MDP with specific state dynamics and specific a class of policies. Throughout the rest of this paper, we refer to the quantum generalizations of classical policies as \emph{classical policies for q-MDPs}, denoting this set of policies as $\Gamma_c$.

Note that any probability measure $\mu$ on a finite set $\sE$ can be written as a density operator on the Hilbert space $\H_{\sE} \coloneqq \rC^{\sE}$ in the following form $\xi_{\mu} \coloneqq \diag\{\mu(e)\}_{e \in \sE}$. Let us denote this injective relation as $T_{c \rightarrow q}: \P(\sE) \rightarrow \D(\H_{\sE})$.  
In view of this quantum representations of classical probability measures on finite sets, one can easily establish that 
\begin{align*}
T_{c \rightarrow q} \circ P(\nu) &= \N \circ  T_{c \rightarrow q}(\nu) \,\,\,\, \forall \nu \in \P(\sX\times\sA) \\
T_{c \rightarrow q} \circ \gamma_t(\mu) &= \gamma_t \circ  T_{c \rightarrow q}(\mu) \,\,\,\, \forall \mu \in \P(\sX), 
\end{align*}
where in the last equation, the first $\gamma_t$ is a classical channel and the second one is a quantum channel. Moreover, we have
$$
C(\nu) = C \circ T_{c \rightarrow q}(\nu) \,\,\,\, \forall \nu \in \P(\sX\times\sA),
$$
where the first $C$ is a function on $\P(\sX\times\sA)$ defined in Section~\ref{sec2new-1} and the second one is a function on $\D(\H_{\sX}\otimes\H_{\sA})$ defined in Section~\ref{sec2new-2}. Hence, d-MDP and its quantum re-formulation are equivalent. 

\section{A Verification Theorem for q-MDPs}\label{sec2new-3}

As in classical MDP theory, we initially establish that storing the entire past information is not necessary for optimality. Namely, we show the existence of an optimal quantum Markov policy. To this end, the dynamic programming operator for q-MDP  
$\L_{g}:C_b(\D(\H_{\sX})) \rightarrow C_b(\D(\H_{\sX}))$
can be defined as follows:
\begin{align*}
\L_{g} V(\rho) &\coloneqq \min_{\gamma \in \N_{(\H_{\sX} \rightarrow \H_{\sX} \otimes \H_{\sA})}} \left[ \langle c,\gamma(\rho) \rangle + \beta \, V(\N \circ \gamma(\rho)) \right],
\end{align*}
where $\N_{(\H_{1} \rightarrow \H_{2})}$ and $C_b(\sE)$ denote the collection of quantum channels from $\H_{1}$ to $\H_{2}$ and the set of bounded and continuous functions on the metric space $\sE$, respectively. We note that for any finite dimensional complex Hilbert space $\H$, the set of density operators $\D(\H)$ is endowed with Hilbert-Schmidt norm. Consequently, the set of quantum channels $\N_{(\H_{1} \rightarrow \H_{2})}$ is endowed with the corresponding operator norm. 
It is straightforward to prove that $\L_{g}$ maps continuous function to continuous function on $\D(\H_{\sX})$ \cite[Proposition D.6]{HeLa96} as $\N_{(\H_{\sX} \rightarrow \H_{\sX} \otimes \H_{\sA})}$ is compact \cite[Proposition 2.28]{Wat18} and the function $\langle c,\gamma(\rho) \rangle + \beta \, V(\N \circ \gamma(\rho))$ is continuous in $\gamma$. Moreover, one can also show that the operator $\L_{g}$ is $\beta$-contraction with respect sup-norm on $C_b(\D(\H_{\sX}))$. Indeed, for any $V, \hat V \in C_b(\D(\H_{\sX}))$, we have 
\begin{align*}
&\sup_{\rho \in \D(\H_{\sX})} \left|\L_{g} V(\rho) - \L_{g} \hat V(\rho) \right|\\ &\leq \sup_{\rho \in \D(\H_{\sX}), \gamma \in \N_{(\H_{\sX} \rightarrow \H_{\sX} \otimes \H_{\sA})}} \left| \langle c,\gamma(\rho) \rangle + \beta \, V(\N \circ \gamma(\rho)) - \langle c,\gamma(\rho) \rangle - \beta \, \hat V(\N \circ \gamma(\rho)) \right| \\
&= \beta \,  \sup_{\rho \in \D(\H_{\sX}), \gamma \in \N_{(\H_{\sX} \rightarrow \H_{\sX} \otimes \H_{\sA})}} \left| V(\N \circ \gamma(\rho)) -  \hat V(\N \circ \gamma(\rho)) \right| \\
&\leq \beta \,  \sup_{\rho \in \D(\H_{\sX})} \left|V(\rho) - V(\rho) \right|
\end{align*}
as $\N \circ \gamma(\rho) \in \D(\H_{\sX})$ for any $\rho \in \D(\H_{\sX})$. Therefore, $\L_g$ has a unique fixed point $V_{\mathrm{fg}} \in C_b(\D(\H_{\sX}))$ by Banach fixed point theorem. The following verification theorem establishes the equivalence of $V^*$ and $V_{\mathrm{fg}}$ and gives a sufficient condition for the optimality of a Markov policy. 

\begin{theorem}\label{verification}
We have $V^* = V_{\mathrm{fg}}$. Moreover, 
for any initial density operator $\rho_0 \in \D(\H_{\sX})$, there exists an optimal Markov policy $\bgamma^*=\{\gamma_{t}^*\}$ satisfying the following for all $t\geq 0$:
\begin{align*}
\gamma_t^* &\in \argmin_{\gamma \in \N_{(\H_{\sX} \rightarrow \H_{\sX} \otimes \H_{\sA})}} \left[ \langle c,\gamma(\rho_t^*) \rangle + \beta \, V_{\mathrm{fg}}(\N\circ\gamma(\rho_t^*)) \right], \\
\intertext{where}
\rho_t^* &= \N(\sigma_{t-1}^*) \,\,\,\, t\geq 1, \,\,\,
\rho^*_0 = \rho_0.
\end{align*}
\end{theorem}

\begin{proof}
Fix any admissible quantum policy $\bgamma = \{\gamma_t\}_{t\geq0} \in \Gamma$. Then, the cost function of $\bgamma$ satisfies the following 
\begin{align*}
V(\bgamma,\rho_0) &= \langle c,\gamma_0(\rho_0) \rangle + \beta \, V(\{\gamma_t\}_{t\geq 1},\N \circ \gamma_0(\rho_0)) \\
&\geq \min_{\gamma_0 \in \N_{(\H_{\sX} \rightarrow \H_{\sX} \otimes \H_{\sA})}} \left[ \langle c,\gamma_0(\rho_0) \rangle + \beta \, V(\{\gamma_t\}_{t\geq 1},\N \circ \gamma_0(\rho_0)) \right] \\
&\geq \min_{\gamma_0 \in \N_{(\H_{\sX} \rightarrow \H_{\sX} \otimes \H_{\sA})}} \left[ \langle c,\gamma_0(\rho_0) \rangle + \beta \, V^*(\N \circ \gamma_0(\rho_0)) \right] = \L_g V^*(\rho_0).
\end{align*}
Since above inequality is true for any quantum policy $\bgamma$, we have $V^* \geq \L_g V^*$. Using the latter inequality and the monotonicity of $\L_g$ (i.e., $V \geq \hat V$ implies $\L_g V \geq L_g \hat V$), one can also prove that $V^* \geq \L^n_gV^*$ for any positive integer $n$. By Banach fixed point theorem, we have $\L^n_g V^* \rightarrow V_{\mathrm{fg}}$ in sup-norm as $n\rightarrow \infty$. Hence, $V^* \geq  V_{\mathrm{fg}}$. 

To prove the converse, for any $\rho$, let 
$$
\gamma_{\rho} \in \argmin_{\gamma \in \N_{(\H_{\sX} \rightarrow \H_{\sX} \otimes \H_{\sA})}} \left[ \langle c,\gamma(\rho) \rangle + \beta \, V_{\mathrm{fg}}(\N\circ\gamma(\rho)) \right].
$$
The existence of $\gamma_{\rho}$ for any $\rho$ follows from the facts that $\N_{(\H_{\sX} \rightarrow \H_{\sX} \otimes \H_{\sA})}$ is compact \cite[Proposition 2.28]{Wat18} and the function $\langle c,\gamma(\rho) \rangle + \beta \, V(\N \circ \gamma(\rho))$ is continuous in $\rho$ and $\gamma$. With this construction, we then have
\begin{align*}
V_{\mathrm{fg}}(\rho_0) &=  \langle c,\gamma_{\rho_0}(\rho_0)\rangle + \beta \, V_{\mathrm{fg}}(\N \circ \gamma_{\rho_0}(\rho_0)) \\
&\eqqcolon \langle c,\sigma_0^*  \rangle + \beta \, V_{\mathrm{fg}}(\rho_1^*) \\
&= \langle c,\sigma_0^*  \rangle + \beta \, \left[\langle c,\gamma_{\rho_1^*}(\rho_1^*) \rangle + \beta \, V_{\mathrm{fg}}(\N \circ \gamma_{\rho_1^*}(\rho_1^*)) \right] \\
&\eqqcolon \langle c,\sigma_0^*  \rangle + \beta \, \left[\langle c,\sigma_1^* \rangle + \beta \, V_{\mathrm{fg}}(\rho_2^*) \right] \\
&\phantom{x}\vdots \\
&\eqqcolon \sum_{t=0}^{N-1} \beta^t \, \langle c,\sigma_t^* \rangle + \beta^N \, V_{\mathrm{fg}}(\rho_N^*) \rightarrow V(\{\gamma_t^*\},\rho_0) \,\, \text{as} \,\ N\rightarrow\infty,
\end{align*}
where $\gamma_t^*(\cdot) \coloneqq  \gamma_{\rho_t^*}(\cdot)$, $\rho_t^* = \N(\sigma_{t-1}^*)$ for each $t\geq 1$, and $\rho_0^* = \rho_0$. Hence, $V_{\mathrm{fg}} \geq V^*$. This implies that
$$
V_{\mathrm{fg}} = V(\{\gamma_t^*\},\rho_0) = V^*
$$
which completes the proof.
\end{proof}

As a result of the verification theorem, for the remainder of this paper, we focus solely on quantum Markov policies. Therefore, when referring to an admissible quantum policy, we inherently mean a Markov policy, despite any notational ambiguity.

\section{Comparison with Previous Quantum MDP Models}\label{sec2new-4}

In this section, we compare q-MDPs with QOMDPs introduced in \cite{BaBaAa14} and other related studies \cite{YiYi18,YiFeYi21}\footnote{As the prior research on quantum MDPs exhibits substantial parallels with QOMDPs introduced in \cite{BaBaAa14}, our primary emphasis centers on the comparison between q-MDP and the generalization of QOMDPs for an extensive assessment. This comparative analysis seeks to highlight the generality and flexibility of q-MDPs in contrast to previous studies.}. The generalization of a QOMDP, denoted again as QOMDP with an abuse of notation, can be specified by a tuple  
$$
\left(\D(\H),\sA,\sM, , \T^{d}, \T^{i}, C\right)
$$
where $\D(\H)$ is the state space, $\sA$ is the finite action space, $\sM$ is the finite observation space, $\T^d = \{\N_{a,d}: a \in \sA\}$ is the set of (divisible) quantum channels that give the measurement to be performed on the state given the action, $\T^{i} = \{\N_{a,i}: a \in \sA\}$ is the set of (indivisible) quantum channels that is performed on the state given the action after the measurement, and $C:\D(\H) \times \sA \rightarrow \R$ is the cost of the current state density operator for a given action.

In this model, given the current state information $\rho \in \D(\H)$, agent picks an action $a \in \sA$ using some policy $\gamma:\D(\H) \rightarrow \sA$ and applies first the corresponding measurement super-operator $\N_{a,d} \in \T^{d}$ to $\rho$ with possible outcomes in space $\sM$. Hence, it is assumed that each $\N_{a,d} \in \T^d$ has $|\sM|$ Kraus operators; that is,
$$
\N_{a,d}(\rho) = \sum_{m \in \sM} N_{a,d}^m \rho N_{a,d}^{m,\dag}. 
$$
If the observation $m$ is received after the application of $\N_{a,d}$, the state evolves to the intermediate state 
$$
\rho_{1/2+} = \frac{N_{a,d}^m \rho N_{a,d}^{m,\dag}}{\tr(N_{a,d}^m \rho N_{a,d}^{m,\dag})}.
$$
Then, the next state of the system can be obtained by applying the indivisible super-operator $\N_{a,i} \in \T^{i}$ to $\rho_{1/2+}$
$$
\rho_{+} = \N_{a,i}(\rho_{1/2+}) = \sum_{m \in \sM} N_{a,i}^m \rho_{1/2+} N_{a,i}^{m,\dag}, 
$$
where $\N_{a,i}(\cdot) = \sum_{m \in \sM} N_{a,i}^m (\cdot) N_{a,i}^{m,\dag}$ is the Kraus representation of the chosen quantum-channel $\N_{a,i}$. Finally, the following cost is paid 
$
C(\rho,a). 
$
Note that the observation $m$ is received with probability $\tr(N_{a,d}^m \rho N_{a,d}^{m,\dag})$. Consequently, the state dynamics evolve randomly in this model. Let $\cE[\cdot]$ represent the expectation operator that accounts for this randomness. Then, the objective is to minimize the discounted cost
$$
J(\gamma,\rho_0) = \sum_{t=0}^{\infty} \beta^t \, \cE[C(\rho_t,a_t)]
$$
by properly selecting the policy $\gamma: \D(\H) \rightarrow \sA$.


Note that the state dynamics of QOMDP can be written as follows
$$
\rho_{t+1} = f(\rho_t,a_t,m_t) 
$$
where 
$$f(\rho,a,m) \coloneqq \sum_{\hat m \in \sM} \frac{N_{a,i}^{\hat m} N_{a,d}^m \rho N_{a,d}^{m,\dag}N_{a,i}^{\hat m,\dag}}{\tr(N_{a,d}^m \rho N_{a,d}^{m,\dag})} $$
and the random variable $m_t$ takes values in $\sM$ and has the following distribution $\rP_{\mathsf{noise}}(m_t = m|\rho,a) = \tr(N_{a,d}^m \rho N_{a,d}^{m,\dag})$. Hence, if we view $m_t$ as a noise, QOMDP is equivalent to a classical Markov decision processes, denoted as c-QOMDP, with the following components 
$$
(\sX,\sA,p,c)
$$
where $\sX \coloneqq \D(\H)$, $c = C$, and
\begin{align*}
p(\cdot|\rho,a) &\coloneqq \sum_{m \in \sM} \delta_{f(\rho,a,m)}(\cdot) \, \rP_{\mathsf{noise}}(m|\rho,a).
\end{align*}
Therefore, classical methodologies like dynamic programming and linear programming are applicable, without any modifications to QOMDPs. Unfortunately, in q-MDPs, direct application of classical techniques is not feasible; they need to be adapted to suit our setup. This is true especially for linear programming formulation.

Moreover, c-QOMDP could be transformed into a deterministic MDP by lifting the state and action spaces to the set of probability measures as established previously. Then, this formulation could be further transformed into q-MDP as done in Section~\ref{classical-q}. Hence, we can view the QOMDP as a specialized instance of q-MDP. However, a complication arises: this transformation demands an infinite-dimensional Hilbert space for the state process, a topic beyond the scope of this paper since in Section~\ref{classical-q}, we only consider the quantum transformation of deterministic reduction of finite state and action MDPs into q-MDPs.

However, we can address this complication by approximating the c-QOMDP with finite state MDPs as established in \cite{SaYuLi17}. Let $d_{\sX}$ denote the metric on $\sX$ induced by Hilbert-Schmidt norm. Since $\sX$ is compact and thus totally bounded, one can find a sequence $\bigl(\{x_{n,i}\}_{i=1}^{k_n}\bigr)_{n\geq1}$ of finite grids in $\sX$ such that for all $n$,
\begin{align}
\min_{i\in\{1,\ldots,k_n\}} d_{\sX}(x,x_{n,i}) < 1/n \text{ for all } x \in \sX. \nonumber
\end{align}
The finite grid $\{x_{n,i}\}_{i=1}^{k_n}$ is called an $1/n$-net in $\sX$. Let $\sX_n \coloneqq \{x_{n,1},\ldots,x_{n,k_n}\}$ and define function $Q_n$ mapping $\sX$ to $\sX_n$ by
\begin{align}
Q_n(x) \coloneqq \argmin_{x_{n,i} \in \sX_n} d_{\sX}(x,x_{n,i}),\nonumber
\end{align}
where ties are broken so that $Q_n$ is measurable. For each $n$, $Q_n$ induces a partition $\{\S_{n,i}\}_{i=1}^{k_n}$ of the state space $\sX$ given by
\begin{align}
\S_{n,i} = \{x \in \sX: Q_n(x)=x_{n,i}\}, \nonumber
\end{align}
with diameter $\diam(\S_{n,i}) \coloneqq \sup_{x,y\in\S_{n,i}} d_{\sX}(x,y) < 2/n$. Let $\{\nu_n\}$ be a sequence of
probability measures on $\sX$ satisfying
\begin{align}
\nu_n(\S_{n,i}) > 0 \text{  for all  } i,n.  \label{compact:numeas}
\end{align}
We let $\nu_{n,i}$ be the restriction of $\nu_n$ to $\S_{n,i}$ defined by
\begin{align}
\nu_{n,i}(\,\cdot\,) \coloneqq \frac{\nu_n(\,\cdot\,)}{\nu_n(\S_{n,i})}. \nonumber
\end{align}
The measures $\nu_{n,i}$ will be used to define a sequence of finite-state MDPs, denoted as c-QOMDP$_{n}$ ($n\geq1$), to approximate the c-QOMDP. To this end, for each $n$ define the one-stage cost function $c_n: \sX_n\times\sA \rightarrow \R$ and the transition probability $p_n$ on $\sX_n$ given $\sX_n\times\sA$ by
\begin{align}
c_n(x_{n,i},a) &\coloneqq \int_{\S_{n,i}} c(x,a) \nu_{n,i}(dx) \nonumber \\
p_n(\,\cdot\,|x_{n,i},a) &\coloneqq \int_{\S_{n,i}} Q_n \ast p(\,\cdot\,|x,a) \nu_{n,i}(dx) \nonumber
\end{align}
where $Q_n\ast p(\,\cdot\,|x,a) \in \P(\sX_n)$ is the pushforward of the measure $p(\,\cdot\,|x,a)$ with respect to $Q_n$; that is,
\begin{align}
Q_n\ast p(x_{n,j}|x,a) = p\bigl(\S_{n,j}|x,a\bigr) \nonumber
\end{align}
for all $x_{n,j} \in \sX_n$. For each $n$, we define c-QOMDP$_n$ as a Markov decision process with the following components: $\sX_n$ is the state space, $\sA$ is the action space, $p_n$ is the transition probability and $c_n$ is the one-stage cost function. The following result establishes that as $n \rightarrow \infty$, we have c-QOMDP$_n \rightarrow$ c-QOMDP. To this end, let $f_n^{*}:\sX_n \rightarrow \sA$ be the optimal deterministic policy of c-QOMDP$_n$. We extend $f_n^{*}$ to $\sX$ by letting $ {\hat f}_n(x) = f_n^* \circ Q_n(x)$.

\begin{theorem}\label{compact:mainthm1}\cite[Theorem 2.2]{SaYuLi17}
The discounted cost of the policy ${\hat f}_n$, obtained by extending the optimal policy $f_n^{*}$ of c-QOMDP$_n$ to $\sX$,
converges to the optimal value function $J^{*}$ of the c-QOMDP 
\begin{align}
\lim_{n\rightarrow\infty} \| J({\hat f}_n,\,\cdot\,) - J^* \|_{\infty} = 0. \nonumber
\end{align}
Hence, to find a near optimal policy for the c-QOMDP, or equivalently for the QOMDP, it is sufficient to compute the optimal policy of QOMDP$_n$ for sufficiently large $n$, and then extend this policy to the original state space.
\end{theorem}

Considering Theorem~\ref{compact:mainthm1}, we can approximate any c-QOMDP (or equivalently any QOMDP) using finite-state c-QOMDP$_n$'s. Since any finite state c-QOMDP$_n$ can be integrated into our q-MDP model as shown in Section~\ref{classical-q}, we can approximate any QOMDP using q-MDPs equipped with classical policies.  This implies that we can assert, at least in an approximate sense, that q-MDPs are strictly more general than QOMDPs.

\subsection{A Special QOMDP: Finite-Action Approximations of q-MDPs}\label{sec2new-5}

In this section, we present a finite-action approximation of q-MDPs by formulating it as a specific instance of QOMDP. Following this, we illustrate its convergence to the original q-MDP.

To simplify the notation, we define $\sU \coloneqq \N_{(\H_{\sX} \rightarrow \H_{\sX} \otimes \H_{\sA})}$. Let $d_{\sU}$ denote the metric on $\sU$ induced by the operator norm. Since $\sU$ is compact, one can find a sequence of finite subsets $\{\Lambda_n\}$ of $\sU$ such that for all $n$
\begin{align}
\min_{\hat{\gamma}\in\Lambda_n} d_{\sU}(\gamma,\hat{\gamma}) < 1/n \text{  } \text{for all $\gamma\in\sU$}. \nonumber
\end{align}
We define the finite-action approximation of q-MDP, denoted as q-MDP$_{n}$, as the QOMDP with the following components: 
$$
\left(\D(\H_{\sX}),\sA,\sM, \T^{d}, \T^{i}, C\right)
$$
where $|\sA| = |\Lambda_n|$, $\T^d = \{\Id: a \in \sA\}$, $\T^{i} = \{\N \circ \N_{a,i}: a \in \sA, \,\, \N_{a,i} \in \Lambda_n\}$, 
$$C(\rho,a) = \tr\left(\N_{a,i}(\rho) c\right),$$ 
and $|\sM|$ is the number of maximum Kraus operators in the representations of quantum channels in $\T^{i}$. Here, since $\T^d = \{\Id: a \in \sA\}$ contains only the identity operator, there is no measurement performed on the system state. Hence, the state dynamics is deterministic. 

The following result establishes that as $n \rightarrow \infty$, q-MDP$_n \rightarrow$ q-MDP. To this end, let $\gamma_n^{*}:\D(\H_{\sX}) \rightarrow \sA$ be an optimal policy of q-MDP$_n$. Let $\{\rho^*_t\}$ be the sequence of state density operators of q-MDP$_n$ under this optimal policy. For each $t$, we define  
$$
\hat{\gamma}_{t,n}(\cdot) = \N_{a^*_t,i}(\cdot)
$$
where $\gamma_n^*(\rho^*_t) = a^*_t$. Hence $\hat{\bgamma} =\{\hat{\gamma}_{t,n}\} \in \Gamma_m$.

\begin{theorem}\label{compact:mainthm3}
The discounted cost of the policy ${\hat \gamma}_n$, obtained by using the optimal policy $\gamma_n^{*}$ of q-MDP$_n$ as above,
converges to the optimal value function $V^{*}$ of the q-MDP 
\begin{align}
\lim_{n\rightarrow\infty}  |V({\hat \gamma}_n,\rho_0) - V^*(\rho_0) | = 0. \nonumber
\end{align}
Hence, to find a near optimal policy for the q-MDP, it is sufficient to compute the optimal policy of q-MDP$_n$ for sufficiently large $n$, and then extend this policy to the q-MDP as above.
\end{theorem}

\begin{proof}
By construction of q-MDP$_n$, we have 
$$
V({\hat \gamma}_n,\rho_0) = J(\gamma_n^*,\rho_0).
$$
By dynamic programming principle, $J(\gamma_n^*,\rho_0)$ is the fixed point of the following operator $\L_{n}:C_b(\D(\H_{\sX})) \rightarrow C_b(\D(\H_{\sX}))$, where
\begin{align*}
\L_{n} V(\rho) &\coloneqq \min_{a \in \sA} \left[ C(\rho,a) + \beta \, V(\N \circ \N_{a,i}(\rho)) \right] \\
&= \min_{\gamma \in \Lambda_n} \left[ \langle c,\gamma(\rho) \rangle + \beta \, V(\N \circ \gamma(\rho)) \right].
\end{align*}
This is almost the same as the dynamic programming operator $\L_g$ of q-MDP, whose fixed point it $V^*$. The only distinction is that we here optimizing the right hand side over $\Lambda_n$ instead of $\sU$. Since $\Lambda_n \rightarrow \sU$ as $n \rightarrow \infty$, by \cite[Theorem 3.2]{SaYuLi16}, we have 
$$
\lim_{n\rightarrow\infty}  |V({\hat \gamma}_n,\rho_0) - V^*(\rho_0) | = \lim_{n\rightarrow\infty}  |J(\gamma_n^*,\rho_0) - V^*(\rho_0) | = 0
$$
This completes the proof.
\end{proof}

Using Theorem~\ref{compact:mainthm3}, we can initially approximate any q-MDP with q-MDP$_n$, a specific instance of QOMDP. Then, by reformulating q-MDP$_n$ as a classical MDP (since it is a QOMDP), we can subsequently approximate it using finite-state MDPs, as demonstrated by Theorem~\ref{compact:mainthm1}. Since any finite-state MDP can be represented as a q-MDP with classical policies (see Section~\ref{classical-q}), we can approximate any q-MDP with q-MDPs with classical policies, where the dimensions of the Hilbert spaces in the latter go to infinity as the approximation converges to zero. Hence, even though the original q-MDP has finite-dimensional Hilbert spaces, in the approximating q-MDP with classical policies, the dimension could be significantly higher than that.

\section{Classes of Quantum Policies}
\label{sec1}

We recall that in d-MDP, a policy is a sequence of classical channels $\bgamma=\{\gamma_{t}\}$ from $\sX$ to $\sX\times\sA$ such that for all $t$
\begin{align} 
\gamma_t(\mu) \in \sA(\mu), \,\, \forall \,\mu \in \P(\sX), \label{inv_classical}
\end{align}
where
$\sA(\mu) \coloneqq \left\{\nu \in \P(\sX \times \sA): v(\,\cdot\,\times\sA) = \mu(\,\cdot\,)\right\}.$ 
Note that we do not impose a condition similar to (\ref{inv_classical}) on quantum (Markov) policies. Presently, any quantum channel from $\H_{\sX}$ to $\H_{\sX} \otimes \H_{\sA}$ is allowed as a policy. Accordingly, the algebraic characterization of any policy is only restricted to the following conditions: completely positive and trace preserving super-operators from $\H_{\sX}$ to $\H_{\sX} \otimes \H_{\sA}$. However, we now introduce different quantum variations of the condition (\ref{inv_classical}) and define corresponding classes of policies like open-loop quantum policies and classical-state-preserving closed-loop  quantum policies. We pursue this for two primary reasons. Firstly, we aim to establish a condition akin to (\ref{inv_classical}) within the quantum case to mirror the formulation given in the classical scenario, allowing us to perceive the quantum formulation as an extension of the classical one. Secondly, the original problem formulation is overly broad and hence unmanageable. Indeed, by verification theorem, to compute the optimal value function and optimal quantum policy, we can apply value-iteration algorithm; that is, we start with some function $V_0 \in C_b(\D(\H_{\sX}))$ and repeatedly apply the operator $\L_{g}$ to that function so that $\L_{g}^n V_0 \rightarrow V^*$ as $n\rightarrow\infty$. In order to repeatedly apply $\L_{g}$ and compute the optimal policy once $V^*$ is obtained, it is necessary to solve the following optimization problem at each time step for given $\rho$ and $V$:
$$
\min_{\gamma \in \N_{(\H_{\sX} \rightarrow \H_{\sX} \otimes \H_{\sA})}} \left[ \langle c,\gamma(\rho)\rangle + \beta \, V(\N \circ \gamma(\rho)) \right].
$$
This optimization problem becomes significantly challenging when there is no particular structure on $\mathcal{N}_{(\mathcal{H}_{\sX} \rightarrow \mathcal{H}_{\sX} \otimes \mathcal{H}_{\sA})}$ and $V^*$. Another alternative method to compute nearly optimal quantum policies is to utilize the approximation results established in Section~\ref{sec2new-5}. Note that any q-MDP can be approximated using finite-action q-MDPs, which can be formulated as QOMDPs. Since any QOMDP can be approximated via a classical MDP with a finite state space, we can approximate any q-MDP using finite-state MDPs. In this approximation, the cardinality of the finite state space in the approximate model tends to infinity as the approximation becomes finer. While this is generally not a significant issue, obtaining the approximate finite model requires discretizing the set of density operators in a nearest neighbor sense. However, although this is theoretically possible, it poses computational challenges. Similarly, to obtain a finite-action approximation of q-MDPs, which can be formulated as QOMDPs, we need to discretize the set of quantum channels $\N_{(\H_{\sX} \rightarrow \H_{\sX} \otimes \H_{\sA})}$. While this is again theoretically feasible due to the compactness of the set, it poses significant computational challenges, as we must preserve the structural properties of complete positivity and trace preservation during the discretization process. As a result, from a computational standpoint, the current formulation of the q-MDP problem presents substantial difficulties. Therefore, it becomes essential to impose some structure on the set of quantum policies. 

Overall, we seek a level of generality that encompasses both the classical formulation and the quantum MDP models introduced in prior literature, such as QOMDP and yet tractable. Hence, we opt to introduce specific constraints by mimicking the condition (\ref{inv_classical}) in the classical d-MDP case. 

In d-MDPs, condition (\ref{inv_classical}) can be rephrased as follows: the marginalization channel 
$\Mar(\nu) \coloneqq \nu(\cdot \times \sA),$ 
where $\Mar: \P(\sX\times\sA) \rightarrow \P(\sX)$, is the inverse of the channel $\gamma_t$; that is, 
$\Mar \circ \gamma_t(\mu) = \mu \,\,\,\, \forall \mu \in \P(\sX).$ 
This relationship can be expressed as $\Mar = \Inv(\gamma_t)$, where the inverse operator $\Inv$ acts on classical channels and reverses the effects of the channel.

Motivated by alternative description of (\ref{inv_classical}) above, we introduce a general reversibility condition on the set of admissible quantum policies. Recall that $\N_{(\H_{1} \rightarrow \H_{2})}$ represents the collection of quantum channels from $\D(\H_{1})$ to $\D(\H_{2})$. Then, the operator 
\begin{align*}
\Inv:\N_{(\H_{\sX} \rightarrow \H_{\sX} \otimes \H_{\sA})} \rightarrow \N_{(\H_{\sX} \otimes \H_{\sA} \rightarrow \H_{\sX})}
\end{align*} 
is defined to invert any channel in $\N_{(\H_{\sX} \rightarrow \H_{\sX} \otimes \H_{\sA})}$ into a channel in $N_{(\H_{\sX} \otimes \H_{\sA} \rightarrow \H_{\sX})}$ in some sense. Although the specific details of this operator are not provided here, upcoming sections propose various alternatives.

Furthermore, we introduce the quantum version of a marginalization channel. The quantum marginalization channel 
\begin{align*}
\Mar:\D(\H_{\sX} \otimes \H_{\sA}) \rightarrow \D(\H_{\sX})
\end{align*}
is a quantum channel that extends the classical probability-based marginalization channel.

Now it is time to define quantum policies that satisfy a reversibility condition similar to (\ref{inv_classical}).  

\begin{definition}\label{admissible}
For an operator $\Inv:\N_{(\H_{\sX} \rightarrow \H_{\sX} \otimes \H_{\sA})} \rightarrow \N_{(\H_{\sX} \otimes \H_{\sA} \rightarrow \H_{\sX})}$, an $\Inv$-Markov quantum policy constitutes a sequence of quantum channels $\bgamma=\{\gamma_{t}\}$ from $\H_{\sX}$ to $\H_{\sX} \otimes \H_{\sA}$ such that
$
\Mar = \Inv(\gamma_t) \,\,\, \forall t.
$
\end{definition}

Throughout the rest of this paper, we explore various frameworks for the operator $\Inv$ and the quantum channel $\Mar$ with increasing generality. Initially, we  examine the most intuitive extensions of the classical operations $\Inv$ and $\Mar$ to their quantum counterparts. This exploration naturally leads to the definition of open-loop quantum policies. Before proceeding to the next section, let us formally define the set of classical policies.

\begin{definition}
An admissible classical policy is a sequence quantum channels $\bgamma=\{\gamma_{t}\}$ from $\H_{\sX}$ to $\H_{\sX} \otimes \H_{\sA}$, where, for all $t$, there exists a classical channel $W_t: \P(\sX) \rightarrow \P(\sX\times\sA)$ realized by the stochastic kernel $\pi_t$ from $\sX$ to $\sA$ such that 
$$
\gamma_t(\rho) = \sum_{(x,a) \in \sX \times \sA} \sqrt{\pi_t(a|x)} \, |x,a\rangle \, \langle x| \, \rho |x\rangle \, \langle x,a| \, \sqrt{\pi_t(a|x)}. 
$$
Here, \(\{|x\rangle\}\) represents a fixed orthonormal basis for \(\H_{\sX}\), while \(\{|x,a\rangle\}\) is an orthonormal basis for \(\H_{\sX} \otimes \H_{\sA}\), constructed from the orthonormal basis \(\{|x\rangle\}\) of \(\H_{\sX}\) and a fixed orthonormal basis \(\{|a\rangle\}\) for \(\H_{\sA}\).
This set of policies is denoted by $\Gamma_c$. 
\end{definition}

\subsection{Open-loop Quantum Policies}\label{sec2}

In this section, the first set of descriptions of the $\Inv$ and $\Mar$ operations, naturally extending the classical descriptions to the quantum case, are provided. Firstly, the inverse operator 
$
\Inv:\N_{(\H_{\sX} \rightarrow \H_{\sX} \otimes \H_{\sA})} \rightarrow \N_{(\H_{\sX} \otimes \H_{\sA} \rightarrow \H_{\sX})},
$ 
which naturally extends the inverse operator in classical case, is defined as the canonical inverse of any quantum channel; that is, for any $\gamma \in \N_{(\H_{\sX} \rightarrow \H_{\sX} \otimes \H_{\sA})}$, we define $\Inv(\gamma) \eqqcolon \gamma^{-1} \in \N_{(\H_{\sX} \otimes \H_{\sA} \rightarrow \H_{\sX})}$ as the quantum channel satisfying the following:
$
\gamma^{-1} \circ \gamma(\rho) = \rho, \, \forall \rho \in \D(\H_{\sX}).
$
To differentiate this from other inverse that is introduced later, let us denote this as $\Inv_{\C}$.

The definition of the marginalization channel is guided by the following observation: the quantum version of the classical marginalization operation is the partial trace operator. Hence, we can define the $\Mar$ as the partial trace channel $\tr_{\sA}: \D(\H_{\sX}\otimes\H_{\sA}) \rightarrow \D(\H_{\sX})$.

\begin{definition}
With these specific descriptions for $(\Inv,\Mar)$, an open-loop quantum policy is defined to be a sequence quantum channels $\bgamma=\{\gamma_{t}\}$ from $\H_{\sX}$ to $\H_{\sX} \otimes \H_{\sA}$ such that, for all $t$, 
$
\tr_{\sA} = \Inv_{\C}(\gamma_t), 
$
or equivalently
$
\gamma_t(\rho) \in \sA(\rho), \, \forall \rho \in \D(\H_{\sX}).
$
Let $\Gamma_o$ denote the set of open-loop policies.
\end{definition}

The following structural result clarifies why we refer to such policies as \emph{open-loop}. The proof is given in the appendix. 

\begin{proposition}\label{structure-qmdp}
Let $\gamma: \D(\H_{\sX}) \rightarrow \D(\H_{\sX}\otimes\H_{\sA})$ be a quantum channel. It satisfies the following reversibility condition: $\tr_{\sA}(\gamma(\rho)) = \rho$ for all $\rho \in \D(\H_{\sX})$ if and only if there exists $\xi \in \D(\H_{\sA})$ such that $\gamma(\rho) = \rho  \otimes \xi$ for all $\rho \in \D(\H_{\sX})$.
\end{proposition}

In the literature, channels described in Proposition~\ref{structure-qmdp} are commonly referred to as \emph{appending channels} \cite{Wil13}. Note that above proposition implies that for any policy $\bgamma = \{\gamma_t\}$ satisfying above invertibility condition, there exists a sequence of density operators $\{\pi_t\} \subset \D(\H_{\sA})$ such that $\gamma_t(\rho) = \rho \otimes \pi_t$ for all $t$. Hence, in view of this, one can alternatively identify the policy $\bgamma$ via this sequence $\{\pi_t\}$. 

Note that, due to the above structural result, these quantum policies do not utilize state information and the policy solely relies on the initial density operator. Consequently, we can categorize these policies as \emph{open-loop quantum policies}, drawing an analogy to the concept of open-loop policies in control theory.

\subsection{Classical-state-preserving Closed-loop  Quantum Policies}\label{sec3}

In this section, we define classical-state-preserving closed-loop  quantum policies for q-MDPs, enabling the controller to utilize current state information. We achieve this by relaxing the reversibility condition inherent in the definition of open-loop policies. In other words, we propose an alternative for the operator $\Inv$, resulting in policies that incorporate current state information. Hence, we refer to them as classical-state-preserving closed-loop  policies.

There is another reason for introducing classical-state-preserving closed-loop  policies, or equivalently, adjusting the operator $\Inv$. Let us elaborate on this point as well. Recall that 
in the quantum version of d-MDP, a classical policy is a sequence quantum channels $\bgamma=\{\gamma_{t}\}$ from $\H_{\sX}$ to $\H_{\sX} \otimes \H_{\sA}$, where, for all $t$, there exists a classical channel $W_t: \P(\sX) \rightarrow \P(\sX\times\sA)$ along with the associated stochastic kernel $\pi_t$ from $\sX$ to $\sA$ such that 
$$
\gamma_t(\rho) = \sum_{(x,a) \in \sX \times \sA} \sqrt{\pi_t(a|x)} \, |x,a\rangle \, \langle x| \, \rho |x\rangle \, \langle x,a| \, \sqrt{\pi_t(a|x)}. 
$$
It is evident that $\tr_{\sA} \circ \gamma_t(\rho) = \rho$ does not hold true for all $\rho \in \D(\H_{\sX})$ for such policies; that is $\tr_{\sA} \neq \Inv_{\C}(\gamma_t)$. As a result, the reversibility requirement in open-loop policies is not satisfied by quantum versions of classical policies. This is as expected because classical policies in d-MDPs operate in a classical-state-preserving closed-loop  manner, utilizing the current state information, unlike open-loop quantum policies.

Therefore, the current definition of q-MDPs with open-loop policies does not allow us to perceive them as straightforward quantum generalizations of classical MDPs. Consequently, it becomes necessary to relax the reversibility constraint by using another inverse operator $\Inv$ in place of $\Inv_{\C}$, enabling quantum versions of classical policies to fulfill it as well. This would allow us to consider q-MDPs as a generalization of d-MDPs.

Recall that in the definition of open-loop policies, we require that an admissible policy is a sequence quantum channels $\bgamma=\{\gamma_{t}\}$ from $\H_{\sX}$ to $\H_{\sX} \otimes \H_{\sA}$ such that, for all $t$, partial trace channel is the full reverse of the quantum channel $\gamma_t$; that is, 
$
\tr_{\sA}(\gamma_t(\rho)) = \rho, \,\,\, \forall \rho \in \D(\H_{\sX}).
$
We can relax this condition in the following way and obtain an alternative inverse operator that leads to the definition of classical-state-preserving closed-loop  policies. To this end, recall the fixed orthonormal basis \(\{|x\rangle\}\) for \(\H_{\sX}\) used in the definition of classical policies, and define \(\S = \{|x\rangle \langle x| : x \in \sX\}\).
Then we have the following definition for the alternative inverse operator.

\begin{definition}\label{inverse1}\cite{Shi13}
The inverse operator  
$
\Inv_{\S}:\N_{(\H_{\sX} \rightarrow \H_{\sX} \otimes \H_{\sA})} \rightarrow \N_{(\H_{\sX} \otimes \H_{\sA} \rightarrow \H_{\sX})}
$
is defined as follows: for any $\gamma \in \N_{(\H_{\sX} \rightarrow \H_{\sX} \otimes \H_{\sA})}$, we define $\Inv_{\S}(\gamma) \eqqcolon \gamma^{-1} \in \N_{(\H_{\sX} \otimes \H_{\sA} \rightarrow \H_{\sX})}$ as the quantum channel satisfying 
$
\gamma^{-1} \circ \gamma(\rho) = \rho, \, \forall \rho \in \S.
$
\end{definition}

As the convex closure of $\S$ is a strict subset of $\D(\H_{\sX})$, this requirement is considerably less stringent than the previous reversibility condition in open-loop policies. Introducing this new inverse operator defines a new set of admissible policies termed as classical-state-preserving closed-loop  policies, a terminology that become clear once we establish the structural properties of such policies.

\begin{definition}\label{closed1}
A \emph{classical-state-preserving closed-loop  policy} is a sequence quantum channels $\bgamma=\{\gamma_{t}\}$ from $\H_{\sX}$ to $\H_{\sX} \otimes \H_{\sA}$ such that, for all $t$, $\tr_{\sA}= \Inv_{\S}(\gamma_t)$; that is, 
$
\tr_{\sA}(\gamma_t(\rho)) = \rho, \,\,\, \forall \rho \in \S.
$
\end{definition}

Let $\cC_w$ denote the set of quantum channels from $\H_{\sX}$ to $\H_{\sX}\otimes\H_{\sA}$ satisfying above relaxed reversibility constraint. Define the set of classical-state-preserving closed-loop  quantum policies as follows
$
\Gamma_w \coloneqq \{\bgamma=\{\gamma_{t}\}: \gamma_t \in {\cC}_w \,\, \forall t\}.
$
Hence, $\Gamma_o \subset \Gamma_w$. Additionally, quantum adaptations of classical policies satisfy this looser requirement, which is stated as a separate result. 

\begin{theorem}
We have $\Gamma_c \subset \Gamma_w$.
\end{theorem}

\begin{proof}
Indeed, for a classical policy $\bgamma = \{\gamma_t\}$, for any $t$, $\gamma_t$ is of the following form 
$$
\gamma_t(\rho) = \sum_{(x,a) \in \sX \times \sA} \sqrt{\pi_t(a|x)} \, |x,a\rangle \, \langle x| \, \rho |x\rangle \, \langle x,a| \, \sqrt{\pi_t(a|x)}.
$$
If $\rho = |y\rangle \langle y|$ for some $y \in \sX$, then 
$
\gamma_t(\rho) = \sum_{a \in \sA} \sqrt{\pi_t(a|y)} \, |y,a\rangle  \langle y,a| \, \sqrt{\pi_t(a|y)}.
$
This implies that $\tr_{\sA}(\gamma_t(\rho)) = \sum_{a \in \sA} \sqrt{\pi_t(a|y)} \, |y\rangle  \langle y| \, \sqrt{\pi_t(a|y)} = |y\rangle \langle y| = \rho$; and so, $\tr_{\sA} = \Inv_{\S}(\gamma_t)$. Hence $\Gamma_c \subset \Gamma_w$. 
\end{proof}

As a result, q-MDPs with classical-state-preserving closed-loop  policies is truly a generalization of the set of d-MDPs. We refer to these q-MDPs as qw-MDPs to distinguish it from q-MDPs with open-loop policies.

We now obtain the following structure of classical-state-preserving closed-loop  quantum policies, whose proof is given in the appendix.

\begin{proposition}\label{structure-qwmdp}
Let $\gamma: \D(\H_{\sX}) \rightarrow \D(\H_{\sX}\otimes\H_{\sA})$ be a quantum channel. It is in $\cC_w$ if and only if there exists a collection of vectors $\{|\phi_{x,a} \rangle\}_{(x,a)\in\sX\times\sA}$ in some Hilbert space $\H_L$ with $\dim(\H_L) \leq |\sX|^2|\sA|$ such that $\sum_{a \in \sA} \langle \phi_{x,a},\phi_{x,a} \rangle=1$ for each $x \in \sX$ and 
\begin{align*}
\gamma(|\psi\rangle \langle \psi|) = \sum_{x,y \in \sX} \langle x,\psi \rangle \langle \psi,y \rangle |x\rangle \langle y| \otimes \left(\sum_{a,b \in \sA} \langle \phi_{y,b},\phi_{x,a} \rangle \, |a\rangle \langle b | \right)
\end{align*}
for each pure state $|\psi\rangle \in \H_{\sX}$.
\end{proposition}

Note that policies of the form
\begin{align*}
\gamma(|\psi\rangle \langle \psi|) = \sum_{x,y \in \sX} \langle x,\psi \rangle \langle \psi,y \rangle |x\rangle \langle y| \otimes \left(\sum_{a,b \in \sA} \langle \phi_{y,b},\phi_{x,a} \rangle \, |a\rangle \langle b | \right)
\end{align*}
obviously use the current state information because of the inner products $\langle \phi_{y,b},\phi_{x,a} \rangle$. As a result, we can categorize   these policies as \emph{classical-state-preserving closed-loop  quantum policies}.

\section{Concluding Remarks}

In conclusion, the objective of this work is to introduce a quantum adaptation of classical MDPs. The approach involves a series of transformations: first, classical MDPs are converted into deterministic models by redefining both state and action spaces as probability measures. Following this initial step, probability measures are replaced with density operators, and the traditional state transition channel is substituted with a quantum channel. These modifications collectively give rise to a quantum version of MDPs. Subsequently, our focus shifts to establishing the verification theorem for q-MDPs and its finite action approximation, which can be cast as a special instance of QOMDP. Finally,  we introduce two categories of policies for q-MDPs: open-loop policies and classical-state-preserving closed-loop  policies, and establish structural results for these classes of policies.

\subsection{Upcoming Work and Future Research Directions}

\subsection{Upcoming Work}
In our upcoming work on q-MDPs, we will develop algorithms for open-loop and classical-state-preserving closed-loop  policies that adapt dynamic programming and linear programming principles to q-MDPs. We will also show that any q-MDP can be approximated via q-MDPs with classical-state-preserving closed-loop  policies, though finer approximations may require larger Hilbert spaces.

\begin{itemize} \item[\ding{43}] We begin by developing algorithms for computing optimal policies and value functions for q-MDPs with open-loop policies, establishing dynamic and semi-definite programming formulations. Using the duality between these formulations, we prove the existence of optimal stationary open-loop policies and derive an algorithm for finding them.
\item[\ding{43}] Next, we focus on classical-state-preserving closed-loop  policies, obtaining dynamic programming principles and semi-definite programming formulations. We prove the existence of optimal stationary classical-state-preserving closed-loop  policies and develop an algorithm for their computation. We also demonstrate that any q-MDP can be approximated by q-MDPs with classical-state-preserving closed-loop  policies, though higher-dimensional Hilbert spaces may be needed.
\item[\ding{43}] Additionally, we revisit the deterministic reduction of classical MDPs, as covered in Section 9.2 of \cite{BeSh78}, and present a linear programming formulation for this reduction, based on \cite[Chapter 6]{HeLa96}. While these results may not be new, they are crucial for setting up a similar analysis in the quantum case. 
\end{itemize}

\subsection{Some Future Research Directions}

The first research direction we aim to explore is the potential benefits of introducing quantum policies into the quantum formulation of classical MDPs. Specifically, since classical-state-preserving closed-loop  policies encompass classical policies as a subset, we will investigate how the use of classical-state-preserving closed-loop  policies can enhance the optimal cost. This research is intricately connected to nonlocal games and Bell inequalities in quantum information theory \cite{Sca19}, as these problems seek to identify the cost structures that allow quantum policies to outperform classical ones in multi-agent decision scenarios.

Another research direction we are interested in developing the mean-field version of q-MDPs. In this framework, both the state dynamics and the cost function will be influenced by the state density operator, introducing non-linear effects. However, a key challenge in this generalization is determining how to integrate these effects into the quantum channel associated with the state dynamics and the one-stage cost, without compromising the integrity of the quantum formalism. This generalization presents a fascinating and challenging area for future research.

\section{Appendix}\label{appendix}

\subsection{Proof of Proposition~\ref{structure-qmdp}}\label{prop-open}
First of all, partial trace $\tr_{\sA}: \D(\H_{\sX} \otimes \H_{\sA}) \rightarrow \D(\H_{\sX})$ is a quantum channel with the following Kraus representation \cite[Section 4.6.2]{Wil13}
$$
\tr_{\sA}(\sigma) = \sum_{a \in \sA} (\Id \otimes \langle a |) \, \sigma \, (\Id \otimes | a \rangle) \eqqcolon \sum_{a \in \sA} K_a \, \sigma \, K_a^{\dag},
$$
where $\{|a\rangle\}$ is an orthonormal basis for $\H_{\sA}$ and $\Id$ is the identity operator. These Kraus operators act in the following way
\begin{align*}
K_a(|\psi\rangle \otimes |\xi\rangle) \coloneqq \langle a | \xi \rangle \, |\psi\rangle, \,\, K_a^{\dag} |\psi\rangle \coloneqq |\psi\rangle \otimes |a\rangle.
\end{align*}
Let $\gamma: \D(\H_{\sX}) \rightarrow \D(\H_{\sX}\otimes\H_{\sA})$ be a quantum channel such that $\tr_{\sA}(\gamma(\rho)) = \rho$ for all $\rho \in \D(\H_{\sX})$. Hence partial trace is the reverse of the quantum channel $\gamma$. Suppose that $\gamma$ has the following Kraus representation \cite[Theorem 4.4.1]{Wil13}
$
\gamma(\rho) = \sum_{l \in L} K_l \, \rho \, K_l^{\dag},
$
where $K_l \in \cL(\H_{\sX},\H_{\sX}\otimes\H_{\sA})$ for $l \in L$, $L$ is some finite set, and $\sum_{l \in L} K_l^{\dag} \, K_l = \Id$. Then, for any pure state $|\psi\rangle \langle\psi|$, we have 
\begin{align*}
|\psi\rangle \langle\psi| &= \tr_{\sA}(\gamma(|\psi\rangle \langle\psi|)) \\
&= \tr_{\sA}\left(\sum_{l \in L} K_l \, |\psi\rangle \langle\psi| \, K_l^{\dag}\right) \\
&= \sum_{a \in \sA} (\Id \otimes \langle a |) \left(\sum_{l \in L} K_l \, |\psi\rangle \langle\psi| \, K_l^{\dag}\right) (\Id \otimes |a\rangle) \\
&= \sum_{\substack{a \in \sA \\ l \in L}} (\Id \otimes \langle a |) K_l \, |\psi\rangle \langle\psi| \, K_l^{\dag} (\Id \otimes |a\rangle).
\end{align*}
For any $(a,l) \in \sA\times L$, we define 
$
|\xi_{a,l} \rangle \coloneqq (\Id \otimes \langle a |) K_l \, |\psi\rangle \in \H_{\sX}.
$
Then, we have
$
|\psi\rangle \langle\psi| = \sum_{\substack{a \in \sA \\ l \in L}} |\xi_{a,l} \rangle \langle \xi_{a,l}|.
$
It is known that pure states are extreme points of $\D(\H_{\sX})$ \cite[page 46]{Wat18}. Hence, for all $(a,l) \in \sA\times L$, we must have 
$$
\frac{|\xi_{a,l} \rangle \langle \xi_{a,l}|}{\tr(|\xi_{a,l} \rangle \langle \xi_{a,l}|)} = |\psi\rangle \langle\psi|,
$$
for all $|\psi \rangle \in \H_{\sX}$. This implies that for all $(a,l) \in \sA\times L$, 
$
(\Id \otimes \langle a |) K_l  = \alpha_{a,l} \, \Id
$
for some $\alpha_{a,l} \in \rC$. For any fix pure state $|\psi\rangle \in \H_{\sX}$, let $\{|\psi\rangle,|\psi_2\rangle, \ldots,|\psi_{|\sX|}\rangle \}$ be an orthonormal basis for $\H_{\sX}$. Then, consider the following orthonormal basis for $\H_{\sX} \otimes \H_{\sA}$: $\{(|\psi\rangle \otimes |a\rangle)_{a \in \sA}, (|\psi_i\rangle \otimes |a\rangle)_{a \in \sA}, i=2,\ldots,|\sX|\}$. Hence, for any $l \in L$, we can write 
$$
K_l |\psi\rangle = \sum_{k=1}^{|\sX|} \sum_{a \in \sA} \beta_{a,k}^{|\psi\rangle} \, |\psi_k\rangle \otimes |a\rangle, 
$$
where $|\psi_1\rangle \coloneqq |\psi\rangle$. However, for each $a \in \sA$, we also have 
$$
(\Id \otimes \langle a |) K_l |\psi\rangle = \sum_{k=1}^{|\sX|} \beta_{a,k}^{|\psi\rangle} \, |\psi_k\rangle = \alpha_{a,l} \, |\psi\rangle.
$$
Since $\{|\psi\rangle,|\psi_2\rangle, \ldots,|\psi_{|\sX|}\rangle \}$ are linearly independent, we have 
$$
\beta_{a,k}^{|\psi\rangle}  = 0 \,\, \text{if} \,\, k \geq 2 \,\, \text{and} \,\, \beta_{a,k}^{|\psi\rangle} = \alpha_{a,l} \,\, \text{if} \,\, k=1.
$$
This implies that 
\begin{align*}
K_l |\psi\rangle &= \sum_{a \in \sA} \beta_{a,1}^{|\psi\rangle} \, |\psi\rangle \otimes |a\rangle \\
&= \sum_{a \in \sA} \alpha_{a,l} \, |\psi\rangle \otimes |a\rangle \\
&= |\psi\rangle \otimes  \sum_{a \in \sA} \alpha_{a,l} \, |a\rangle 
\eqqcolon |\psi\rangle \otimes |\pi_l\rangle.
\end{align*}
Note that $|\pi_l\rangle$ does not depend on $|\psi\rangle$, and so, 
\begin{align*}
\gamma(|\psi\rangle \langle\psi|) &= \sum_{l \in L} K_l |\psi\rangle \langle\psi| K_l^{\dag} \\
&= \sum_{l \in L} |\psi\rangle \otimes |\pi_l\rangle \langle\psi| \otimes \langle\pi_l| \\
&= \sum_{l \in L} |\psi\rangle \langle\psi|  \otimes |\pi_l\rangle \langle\pi_l| \\
&= |\psi\rangle \langle\psi|  \otimes \sum_{l \in L} |\pi_l\rangle \langle\pi_l| \eqqcolon |\psi\rangle \langle\psi|  \otimes \xi,
\end{align*}
where again $\xi$ does not depend on $|\psi\rangle \langle\psi|$. Obviously, $\xi \in \D(\H_{\sA})$ since $\xi = \tr_{\sX}(\gamma(|\psi\rangle \langle\psi|))$. In conclusion, $\gamma(|\psi\rangle \langle\psi|) = |\psi\rangle \langle\psi|  \otimes \xi$ for some fixed $\xi \in \D(\H_{\sA})$. Since pure states are extreme points of $\D(\H_{\sX})$, we also have $\gamma(\rho) = \rho  \otimes \xi$ for all $\rho \in \D(\H_{\sX})$. This indeed completes the proof as the reverse implication is obvious.

\subsection{Proof of Proposition~\ref{structure-qwmdp}}\label{prop-closed}

Recall that partial trace $\tr_{\sA}: \D(\H_{\sX} \otimes \H_{\sA}) \rightarrow \D(\H_{\sX})$ has the following Kraus representation
$$
\tr_{\sA}(\sigma) = \sum_{a \in \sA} (\Id \otimes \langle a |) \, \sigma \, (\Id \otimes | a \rangle) \eqqcolon \sum_{a \in \sA} K_a \, \sigma \, K_a^{\dag},
$$
where $\{|a\rangle\}$ is an orthonormal basis for $\H_{\sA}$. Let $\gamma: \D(\H_{\sX}) \rightarrow \D(\H_{\sX}\otimes\H_{\sA})$ be a quantum channel such that $\tr_{\sA}(\gamma(\rho)) = \rho$ for all $\rho \in \S$. Hence partial trace is the reverse of the quantum channel $\gamma$ with respect to $\S$. Suppose that $\gamma$ has the following Kraus representation \cite[Theorem 4.4.1]{Wil13}
$
\gamma(\rho) = \sum_{l \in L} K_l \, \rho \, K_l^{\dag},
$
where $K_l \in \cL(\H_{\sX},\H_{\sX}\otimes\H_{\sA})$ for $l \in L$, $|L|< |\sX|^2 |\sA|$, and $\sum_{l \in L} K_l^{\dag} \, K_l = \Id$. Then, for any pure state $|x\rangle \langle x| \in \S$, we have 
\begin{align*}
|x\rangle \langle x| &= \tr_{\sA}(\gamma(|x\rangle \langle x|)) \\
&= \tr_{\sA}\left(\sum_{l \in L} K_l \, |x\rangle \langle x| \, K_l^{\dag}\right) \\
&= \sum_{a \in \sA} (\Id \otimes \langle a |) \left(\sum_{l \in L} K_l \, |x\rangle \langle x| \, K_l^{\dag}\right) (\Id \otimes |a\rangle) \\
&= \sum_{\substack{a \in \sA \\ l \in L}} (\Id \otimes \langle a |) K_l \, |x\rangle \langle x|\, K_l^{\dag} (\Id \otimes |a\rangle).
\end{align*}
For any $(a,l,x) \in \sA\times L \times \sX$, we define 
$
|\xi_{a,l}^x \rangle \coloneqq (\Id \otimes \langle a |) K_l \, |x\rangle \in \H_{\sX}.
$
Then, we have
$
|x\rangle \langle x| = \sum_{\substack{a \in \sA \\ l \in L}} |\xi_{a,l}^x \rangle \langle \xi_{a,l}^x|.
$
It is known that pure states are extreme points of $\D(\H_{\sX})$. Hence, for all $(a,l,x) \in \sA\times L \times \sX$, we must have 
$$
\frac{|\xi_{a,l}^x \rangle \langle \xi_{a,l}^x|}{\tr(|\xi_{a,l}^x \rangle \langle \xi_{a,l}^x|)} = |x\rangle \langle x|.
$$
for all $|x \rangle \in \H_{\sX}$. This implies that for all $(a,l,x) \in \sA\times L \times \sX$, 
$
(\Id \otimes \langle a |) K_l |x\rangle  = \alpha_{a,l}^x \, |x\rangle
$
for some $\alpha_{a,l}^x \in \rC$. Note that $\sum_{\substack{a \in \sA \\ l \in L}} |\alpha_{a,l}^x|^2 = \langle x,x\rangle =1$.

Note that for each $l$,  we have
$
K_l |x\rangle = \sum_{y \in \sX} \sum_{b \in \sA} m_{y,b}^x \, |y,b\rangle. 
$
But since $(\Id \otimes \langle a |) K_l |x\rangle  = \alpha_{a,l}^x \, |x\rangle$, for each $a \in \sA$, we must have $m_{y,a}^x = 0$ if $y \neq x$ and $m_{y,a}^x = \alpha_{a,l}^x$ if $y=x$. Hence, we indeed have
$
K_l |x\rangle = \sum_{b \in \sA} \alpha_{b,l}^x \, |x,b\rangle. 
$
Since $\{|x\rangle\}$ is an orthonormal basis for $\H_{\sX}$, for any pure state $|\psi\rangle \in \H_{\sX}$, we have
$
|\psi\rangle = \sum_{x \in \sX} \langle x,\psi \rangle \, |x\rangle.
$
Hence
$
K_l |\psi\rangle = \sum_{\substack{a \in \sA \\ x \in \sX}}  \langle x,\psi \rangle \, \alpha_{a,l}^x \, |x,a\rangle. 
$
This then implies that 
\begin{align*}
\gamma(|\psi\rangle \langle \psi|) &= \sum_{l \in L} K_l |\psi\rangle \langle \psi| K_l^{\dag} \\ 
&= \sum_{x,y \in \sX} \langle x,\psi \rangle \langle \psi,y \rangle |x\rangle \langle y| \otimes \left(\sum_{a,b \in \sA} \Theta_{a,x,b,y} \, |a\rangle \langle b | \right),
\end{align*}
where $\Theta_{a,x,b,y} \coloneqq \sum_{l \in L} \alpha_{a,l}^x (\alpha_{b,l}^y)^*$ and $\alpha^*$ denotes the complex conjugate of $\alpha$. 
Let $\H_L$ be a Hilbert space of dimension $|L|$ with the following orthonormal basis $\{|l\rangle\}_{l \in L}$. Then define the following vectors in $\H_L$ for each $(x,a) \in \sX \times \sA$,
$
|\phi_{x,a}\rangle \coloneqq \sum_{l \in L} \alpha_{a,l}^x \, |l\rangle.
$
We have $\Theta_{a,x,b,y} = \langle \phi_{y,b},\phi_{x,a} \rangle$. Hence we can write 
\begin{align*}
\gamma(|\psi\rangle \langle \psi|) = \sum_{x,y \in \sX} \langle x,\psi \rangle \langle \psi,y \rangle |x\rangle \langle y| \otimes \left(\sum_{a,b \in \sA} \langle \phi_{y,b},\phi_{x,a} \rangle \, |a\rangle \langle b | \right).
\end{align*}
Note that for each $x \in \sX$, we also have 
$$\sum_{a \in \sA} \langle \phi_{x,a},\phi_{x,a} \rangle = \sum_{\substack{a \in \sA \\ l \in L}} |\alpha_{a,l}^x|^2 =1.$$
Hence, the proof of the 'if' part is complete.

For the proof of the converse, let $\gamma$ be a super-operator from $\H_{\sX}$ to $\H_{\sX} \otimes \H_{\sA}$ of the following form 
\begin{align*}
\gamma(|\psi\rangle \langle \psi|) = \sum_{x,y \in \sX} \langle x,\psi \rangle \langle \psi,y \rangle |x\rangle \langle y| \otimes \left(\sum_{a,b \in \sA} \langle \phi_{y,b},\phi_{x,a} \rangle \, |a\rangle \langle b | \right),
\end{align*}
where $\{|\phi_{x,a} \rangle\}_{(x,a)\in\sX\times\sA}$ is a collection of vectors in some Hilbert space $\H_L$ such that $\sum_{a \in \sA} \langle \phi_{x,a},\phi_{x,a} \rangle=1$ for each $x \in \sX$. We first prove that $\gamma$ is a quantum channel.

Define the isometry $V:\H_{\sX} \rightarrow H_{\sX} \otimes \H_{\sA} \otimes \H_L$ as follows
$$
V: |x\rangle \mapsto |x\rangle \otimes \sum_{a \in \sA} |a\rangle \otimes |\phi_{x,a}\rangle.
$$
Indeed, for each $|x\rangle,|y\rangle \in \H_{\sX}$, we have $\langle x| V^{\dag} V |x\rangle = \sum_{a \in \sX} \langle \phi_{x,a},\phi_{x,a} \rangle =1$ and $\langle x| V^{\dag} V |y\rangle = 0$ if $x \neq y$. Hence $V^{\dag} V = \Id$, and so, $V$ is indeed an isometry. Hence the following super-operator is a quantum channel:
$
\N_V(\rho) = V \rho V^{\dag}.
$
Note that 
\begin{align*}
\tr_L(\N_V(|\psi\rangle \langle \psi|)) &= \tr_L\left(V\left(\sum_{x} \langle x,\psi \rangle |x\rangle \right)\left(\sum_{x} \langle \psi,x \rangle \langle x| \right)V^{\dag}\right) \\
&= \sum_{x,y} \langle x,\psi \rangle \langle \psi,y \rangle \tr_L\left(V |x\rangle  \langle y| V^{\dag}\right) \\
&=  \sum_{x,y} \langle x,\psi \rangle \langle \psi,y \rangle \tr_L\left(|x\rangle \otimes \sum_{a \in \sA} |a\rangle \otimes |\phi_{x,a}\rangle \, \langle y| \otimes \sum_{b \in \sA} \langle b| \otimes \langle\phi_{y,b}|\right) \\
&=  \sum_{x,y} \langle x,\psi \rangle \langle \psi,y \rangle |x\rangle \langle y| \otimes  \sum_{a,b \in \sA}   \langle\phi_{y,b},\phi_{x,a}\rangle |a\rangle \langle b| \eqqcolon \gamma(|\psi\rangle \langle \psi|) 
\end{align*}
This implies that quantum channel $\N_V$ is a Stinespring representation of the super-operator $\gamma$ \cite[Definition 5.2.1]{Wil13}. Therefore, $\gamma$ is also a quantum channel. 

Let us now check the relaxed reversibility condition. Indeed, for any $|z\rangle \langle z| \in \S$, we have
$$
\gamma(|z\rangle \langle z|) = |z\rangle \langle z| \otimes \left(\sum_{a,b \in \sA} \langle \phi_{z,b},\phi_{z,a} \rangle \, |a\rangle \langle b | \right).
$$
Hence 
$
\tr_{\sA}(\gamma(|z\rangle \langle z|) ) = |z\rangle \langle z|  \left(\sum_{a \in \sA} \langle \phi_{z,a},\phi_{z,a} \rangle \right) = |z\rangle \langle z|. 
$
This means that $\gamma$ satisfies relaxed reversibility condition. This completes the proof.


\end{document}